\newif\ifshortver
\shortvertrue
\shortverfalse
 
\ifshortver
\documentclass[conference]{IEEEtran}
\else
\documentclass[onecolumn, english]{IEEEtran}
\fi

\IEEEoverridecommandlockouts

\usepackage[hidelinks   ]{hyperref} 
\usepackage{algorithm}
\usepackage[noend]{algpseudocode}

\algrenewcommand\algorithmicindent{5mm}%

\usepackage{amsmath}
\allowdisplaybreaks
\usepackage{amssymb,amsthm}
\usepackage{amsfonts}
\usepackage{cite}

\usepackage{amssymb,amsthm}
\usepackage{amsfonts}
\usepackage{cite}
\usepackage{tikz}
\usepackage{bbm}
\usetikzlibrary{positioning,arrows,shapes,chains,fit,scopes}
\usepackage{pgfplots}
\usepackage{pgfplotstable}	
\usepackage{booktabs}

\usepackage{threeparttable}
\usepackage{makecell}
\usepackage{multirow}

\usepackage{colortbl}
\usepackage{color}
\definecolor{bgcolor}{rgb}{0.93,0.99,1}
\definecolor{bgcolor2}{rgb}{0.8,1,0.8}
\definecolor{bgcolor3}{rgb}{0.50,0.90,0.50}
\usepackage{tcolorbox}
\usepackage{pifont}
\definecolor{mydarkgreen}{RGB}{39,130,67}
\definecolor{mydarkred}{RGB}{192,25,25}

\usepackage{blkarray, bigstrut}
 
\allowdisplaybreaks

\theoremstyle{theorem}
\newtheorem{theorem}{Theorem}
\newtheorem{lemma}[theorem]{Lemma}

\newtheorem{proposition}[theorem]{Proposition}
\theoremstyle{definition}
\newtheorem{definition}{Definition}

\theoremstyle{remark}
\newtheorem{remark}{Remark}

\usepackage{lipsum}
\usepackage{mathtools}
\usepackage{cuted}

\title{One-Shot Information Hiding and Compound Wiretap Channels}

\begin{document}

\author{%
Yanxiao Liu, \textit{Member, IEEE} and Cheuk Ting Li, \textit{Member, IEEE}
\thanks{
This work was partially supported by two grants from the Research Grants Council of the Hong Kong Special Administrative Region, China [Project No.s: CUHK 24205621 (ECS), CUHK 14209823 (GRF)].    

This paper was presented in part at the 2024 IEEE Information Theory Workshop (ITW 2024).

Yanxiao Liu was with the Department of Information Engineering, The Chinese University of Hong Kong, Hong Kong SAR, China, and this work was done during his time there. 
He is now with the Department of Electrical and Electronic Engineering, Imperial College London, London, UK.
Email: yanxiaoliu@link.cuhk.edu.hk

Cheuk Ting Li is with the Department of Information Engineering and the Institute of Theoretical Computer Science and Communications, The Chinese University of Hong
Kong, Hong Kong SAR of China. Email:  ctli@ie.cuhk.edu.hk. 
}
}

\maketitle

\begin{abstract}

In recent years, due to the growing reliance on large amounts of data that are communicated, analyzed, and utilized, which inherently contain sensitive and personal information, secrecy and privacy in communication have become increasingly important. 
We present nonasymptotic information-theoretic analyses of two fundamental secrecy problems under channel uncertainties: the information hiding problem and the compound wiretap channel. 
The former admits a game-theoretic formulation, where one party (an information hider and a decoder) seeks to embed secret messages into a host signal for later reconstruction, while the opposing party (an attacker) attempts to remove or degrade the embedded information.
The latter generalizes Wyner's wiretap channel by allowing multiple potential channel states. 
The information hiding problem concerns \emph{active} attacks during data transmission, while the compound wiretap channel addresses \emph{passive} eavesdropping and information leakage. 
The two problems, both of which consider channels with uncertainties, can be studied under a unified framework by utilizing a covering argument and the Poisson matching lemma by Li and Anantharam. 
We derive novel one-shot achievability results for both problems that are applicable to any source distribution and any class of channels (not necessarily memoryless or ergodic), and that apply to both discrete and continuous cases. 
We also show that existing asymptotic results on both problems can be recovered by applying our results to discrete memoryless channels.
\end{abstract}


\section{Introduction}
\label{sce::intro}

In data science and wireless communication, secrecy and privacy are becoming increasingly important due to the growing dependence on large amounts of data being communicated, analyzed, and utilized, which inherently contain sensitive and personal information and, therefore, should be well protected from leakage. 
Over the past few decades, the fundamental information-theoretic limits of various secrecy and privacy problems have been extensively studied. 
During information transmission, there are two types of concerns regarding information leakage: \emph{active} attacks and \emph{passive} wiretapping. 
In this paper, we consider two fundamental problems in information theory that address these two concerns, respectively: the information hiding problem~\cite{moulin2003information} and the compound wiretap channel~\cite{liang2009compound}.

For the information hiding problem~\cite{moulin2003information}, it can be formulated  as a communication system from a game-theoretic perspective, where an encoder-decoder team seeks to transmit a confidential message \emph{embedded} in a host data source, while the opposing side is an attacker, modeled as a noisy channel, attempts to destroy or degrade the message. 
The information-theoretic limits of different variations of information hiding have been extensively investigated over the past two decades~\cite{moulin2003information, somekh2003error, somekh2004capacity, cohen2002gaussian}, due to its wide range of applications, including watermarking, fingerprinting, steganography, and copyright protection. 
Existing analyses of information hiding problems borrow techniques from various fields, including wireless communication, signal processing, cryptography, and game theory. 
For the compound wiretap channel, \cite{liang2009compound} modeled the problem as a generalization of Wyner's conventional wiretap channel~\cite{wyner1975wire}, where the communication channel can take multiple potential states.
The objective is to ensure reliable transmission and minimize information leakage regardless of which state occurs.
This model is more general and better suited to practical wireless communication scenarios where the transmitter may not have knowledge of the channel conditions or where channel characteristics change rapidly, yet communication performance and security must still be guaranteed.

In all existing studies on these two problems, the information-theoretic limits have been analyzed in the \emph{asymptotic} regime, assuming that the signal has a blocklength approaching infinity.  
However, this assumption does not hold in practice, as packets have bounded lengths, which can be quite short in certain applications~\cite{ji2018ultra}. 
Over the past decade, \emph{finite-blocklength information theory} has been extensively studied~\cite{kostina2013lossy, tan2013dispersions, polyanskiy2010channel, hayashi2009information}, leading to the derivation of nonasymptotic limits, where the law of large numbers does not apply and conventional typicality-based tools are inapplicable.  
More generally, we are interested in \emph{one-shot} achievability results, where the channel or source is arbitrary and used only \emph{once}.
Various one-shot coding techniques have been proposed~\cite{verdu2012non, yassaee2013technique, li2021unified}, yielding one-shot bounds that recover existing (first-order and second-order) asymptotic results when applied to memoryless sources or channels. See Section~\ref{subsec::review_oneshot} for a review.

In this paper, which extends the conference version~\cite{liu2024hiding},\footnote{
The conference version~\cite{liu2024hiding} studies the information hiding setting only. In comparison, this journal version not only provides more detailed discussions and generalizations of the information hiding problem, but also extends the framework to the compound wiretap channel and derives novel one-shot achievability results, which were not covered in~\cite{liu2024hiding}.} we study the one-shot achievability results of the information hiding problem and the compound wiretap channel. 
Compared to the asymptotic information hiding~\cite{moulin2003information}, we remove the assumption that the decoder knows the attacker's strategy, considering instead that the only available knowledge is that the attack channel belongs to a set (which may have infinite cardinality).
For both problems, our goal is to provide \emph{distributionally robust} coding strategies (also see~\cite{malik2024distributionally}) due to the channel uncertainty, and we hence study them under a unified framework.
We briefly summarize our contributions as follows.

\begin{itemize}
    \item We derive novel one-shot achievability results for the information hiding problem~\cite{moulin2003information} and the compound wiretap channel~\cite{liang2009compound}, under a unified framework that is designed for secrecy problem with channel uncertainties. 
    
    \item For both problems, most of the existing asymptotic analyses assume the decoder knows the channels.\footnote{This assumption can be reasonable when the blocklength is large, but should be dropped in the one-shot scenarios. See Section~\ref{subsec::disc} for discussions.} 
    We derive one-shot results without such an assumption, by utilizing the Poisson matching lemma~\cite{li2021unified} and a covering argument~\cite{blackwell1959capacity}. 

    \item We show the one-shot results recover existing asymptotic bounds when applied to memoryless channels (possibly subject to distortion constraints), hence providing alternative proofs that can be simpler. 

    \item Additionally, unlike~\cite{moulin2003information, liang2009compound}, our results apply to both continuous and discrete cases.
\end{itemize}

The paper is organized as follows.
We begin with a literature review on information hiding, watermarking, compound wiretap channels, and one-shot information theory in Section~\ref{sec::review}.
Next, we present the one-shot information hiding problem in Section~\ref{sec::generalized_info_hiding} and recover existing asymptotic results in Section~\ref{sec::recovery_hiding}.
Within the same framework, we study the one-shot compound wiretap channel in Section~\ref{sec::wiretap}. 
Finally, in Section~\ref{sec::watermarking_AI}, we discuss how our analysis on information hiding can benefit the design of better AI-based watermarking tools, which are becoming increasingly important in the era of generative models.

\section{Related Literature}
\label{sec::review}

\subsection{Information Hiding}
\label{subsec::review_infohiding}

The information hiding problem has been studied since~\cite{cohen2002gaussian, moulin2003information, somekh2003error, somekh2004capacity}, due to its wide range of applications, including watermarking, fingerprinting, audio/image/video processing, copyright protection, and steganography.
The goal is to \emph{hide} a message into some host signal (by introducing a certain level of distortion), so that the message can be correctly reconstructed after suffering \emph{attacks} (which introduce another level of distortion).
This problem was modeled as a communication problem, and asymptotic information-theoretic capacity was derived in~\cite{moulin2003information}.
In general, information hiding is closely related to the Gelfand–Pinsker problem~\cite{gelfand1980coding,heegard1983capacity}, and various extensions have been studied, e.g., the case where side information is available to the encoder, decoder, and adversary~\cite{moulin2007capacity}, and the case where the decoder has rate-limited side information~\cite{steinberg2008coding}. 
See~\cite{barron2003duality} for its duality with the Wyner–Ziv problem, and~\cite{keshet2008channel} for a comprehensive survey.
We discuss its applications and related settings with different objectives as follows.

\subsubsection{\textbf{Watermarking, Fingerprinting, and Steganography}} 

The setting in~\cite{moulin2003information} can be viewed as \emph{public watermarking}~\cite{somekh2004capacity}, where the host signal is available only at the encoder. In contrast, when it is also available at the decoder, \emph{private watermarking} has been studied in~\cite{somekh2003error, cox1997secure}. In the Gaussian case, public and private watermarking have the same capacity~\cite{cohen2002gaussian}, but this is not true in general. Watermarking problems consider messages containing personal identification information to be protected from attacks, but secrecy is not always required. In comparison, \emph{digital fingerprinting}~\cite{boneh1998collusion, moulin2003information} embeds fingerprints into the host data to uniquely identify users for tracing illegal data usage, which can be more challenging due to potential collusion. A provably good data embedding strategy was introduced by~\cite{chen2001quantization}. Random coding error exponents have been investigated in these problems~\cite{moulin2007capacity, merhav2000random, moulin2008universal}. 
In~\cite[Sec. VII.C]{moulin2003information} it was indicated that information hiding was applicable to steganography, and this was later studied by~\cite{wang2008perfectly} for the capacity of perfectly secure steganographic systems. 
See \cite{guan2022double} for trellis codes and \cite{li2020designing} for polar codes, which are used in steganographic code designs. 
See \cite{}

\subsubsection{\textbf{Host, Stegotext, and Reversibility}} 

In the conventional information hiding~\cite{moulin2003information}, the message is embedded into host data by producing an encoded signal (``stegotext''), with the goal of recovering the message only. Other objectives have been considered later, such as conveying the host~\cite{kim2008state} or reconstructing the 
stegotext~\cite{grover2015information, xu2023information}. \emph{Reversible} information embedding has also been investigated~\cite{kalker2002capacity, steinberg2006reversible, sumszyk2009information}, where the host signal needs to be decoded. However, this can incur a high cost when the host has high entropy~\cite{kalker2002capacity}, making perfect reversibility even impossible for continuous host signals~\cite{sumszyk2009information}. Nevertheless, in practice, the goal is often to enable retransmission of the stegotext, and codes for stegotext recovery have been studied~\cite{grover2015information, xu2023information}. 
For this setting, single-letter capacity-distortion tradeoffs are known only for logarithmic distortion~\cite{kim2008state} and quadratic distortion in the Gaussian case~\cite{sutivong2005channel}. 
We discuss these generalizations on our setting in Section~\ref{subsec::hiding_generalization}.

\subsection{Compound Wiretap Channels}
\label{subsec::cpd_wiretap}

Compound wiretap channels~\cite{liang2009compound} generalize the conventional wiretap channel model~\cite{wyner1975wire} by allowing both the legitimate channel and the eavesdropper's channel to have multiple possible states. The objective is to guarantee reliable and secure signal transmission regardless of which state occurs. This is a practical model for channel uncertainty, where the transmitter may have no knowledge of the channel (due to the dynamic nature of the wireless medium or unavoidable implementation/estimation inaccuracies), but zero performance outage is still required (e.g., for ultra-reliable communications~\cite{ji2018ultra}). 
\cite{liang2009compound} proposed achievable and converse results, with the converse bounds shown to be tight in certain cases by~\cite{bjelakovic2013secrecy}. They also studied the achievable secrecy degrees of freedom (s.d.o.f.) region for a multi-input multi-output (MIMO) model, which was later extended to the case of two confidential messages in~\cite{kobayashi2009compound}. The s.d.o.f. of compound wiretap parallel channels were also studied in~\cite{liu2008secrecy}. See~\cite{ekrem2010gaussian, ekrem2012degraded} for discussions on Gaussian MIMO compound wiretap channels.

In~\cite{liang2009compound, bjelakovic2013secrecy}, the results focused on discrete memoryless channels with a countably finite uncertainty set (i.e., the set from which the exact channel realization is drawn). This was later extended to \emph{arbitrary uncertainty sets}, including continuous alphabets, by~\cite{schaefer2015secrecy}, which is also one of the contributions of our paper. Moreover, \cite{boche2015continuity} showed that the secrecy capacity is a continuous function of the uncertainty set.

\subsection{One-shot Information Theory}
\label{subsec::review_oneshot}

For all the literature discussed in this section so far, the information-theoretic bounds are investigated \emph{asymptotically} in the large blocklength limit, based on the law of large numbers. 
However, this assumption is impractical, as packets have bounded lengths, which can be very short in low-latency communications~\cite{ji2018ultra}. 
Over the past decade, \emph{finite blocklength}~\cite{kostina2013lossy, tan2013dispersions, polyanskiy2010channel} and \emph{one-shot}~\cite{hayashi2009information, verdu2012non, yassaee2013technique, song2016likelihood, watanabe2015nonasymp, li2021unified} information theory have been widely studied, leading to nonasymptotic results. 
In the one-shot setting, we assume the channel or source is used only \emph{once} (i.e., it need not be memoryless or ergodic), and the blocklength is $1$. 
These results are expected to recover existing (first-order and second-order) asymptotic bounds when applied to memoryless channels or sources (e.g., asymptotic channel coding capacity by Shannon~\cite{shannon1948mathematical} can be implied by the one-shot bounds of Feinstein~\cite{feinstein1954new}, Shannon~\cite{shannon1957certain}, or~\cite{hayashi2009information, polyanskiy2010channel}). 
Similar to \cite{liu2025one, liu2025nonasymptotic}, in this paper we utilize the Poisson functional representation \cite{li2018strong} for one-shot results. 
It has also been used for differential privacy \cite{liu2024universal, flamich2026scalable} and watermarked speculative sampling \cite{liu2026watermarkable}.

In this paper, we consider the one-shot scenario and derive achievability results using the \emph{Poisson matching lemma}~\cite{li2021unified}, which has been shown to improve upon previously known one-shot bounds in various settings with simpler analyses~\cite{li2021unified}, with recent applications including hypothesis testing~\cite{guo2024hypothesis}, secret key generation~\cite{hentila2024communication}, unequal message protection~\cite{khisti2024unequal} and oblivious relaying~\cite{liu2025nonasymptotic}. 
It is based on the Poisson functional representation~\cite{li2018strong}, which has also recently been applied to various fields in combination with other techniques, such as neural estimation~\cite{lei2022neural} and differential privacy~\cite{liu2024universal}. 
A refined version of the Poisson matching lemma has been used to analyze general (multi-hop) noisy networks in~\cite{liu2025one}.

\subsection*{Notations}

We assume logarithm and entropy are to the base $2$. 
For a statement $S$, we use $\mathbf{1}\{S\}$ to denote its indicator, i.e., $\mathbf{1}\{S\} = 1$ if $S$ holds, and otherwise $\mathbf{1}\{S\} = 0$.
We use $\delta_{a}$ to denote the degenerate distribution $\mathbf{P}\{X=a\}=1$. 
For two random variables $X,Y$, the information density is defined as $\iota_{X;Y}(x;y) = \log((\mathrm{d} P_{X|Y}(\cdot |y)/ \mathrm{d} P_X)(x))$, where $\mathrm{d} P_{X|Y}(\cdot |y)/ \mathrm{d} P_X$ denotes the Radon-Nikodym derivative. We sometimes omit the subscript and write $\iota(x;y)$ if the random variables are clear from the context. 
The total variation (TV) distance between two distributions $P,Q$ over $\mathcal{X}$ is $\Vert P-Q\Vert_{\mathrm{TV}} := \sup_{A \subseteq \mathcal{X} \; \text{measurable}} |P(A)-Q(A)|$. 
The set of positive integers is denoted by $\mathbb{N}_+$. 
We write $[i] := \{1,\ldots, i\}$ for $i\in\mathbb{N}_+$. 
We write $P \ll Q$ to denote that $P$ is absolutely continuous with respect to $Q$. 
\smallskip

\section{Poisson Matching Lemma}
\label{sec::PML}

In this section, we review the Poisson matching lemma~\cite{li2021unified}, which will be used throughout this paper. 
The techniques in~\cite{moulin2003information, somekh2004capacity}, e.g., tools based on \emph{typical sets}, which have resemblances to Gelfand--Pinsker coding~\cite{gelfand1980coding, heegard1983capacity}, \textcolor{black}{are not suitable for the one-shot setting, since they rely on the law of large numbers to study asymptotic behavior}. 
The Poisson matching lemma has been shown to be useful in proving one-shot achievability results of network information theory~\cite{li2021unified, liu2025one}, see Section~\ref{subsec::disc} for a detailed discussion.

The Poisson matching lemma is rooted in the Poisson functional representation~\cite{li2018strong} that is reviewed as follows. 
Fix a probability distribution $\bar{P}$ over $\mathcal{U}$. 
Let $(T_i)_{i=1,2,\ldots}$ be a Poisson process with rate $1$, i.e., $T_1,T_2-T_1,T_3-T_2 \stackrel{\mathrm{iid}}{\sim} \mathrm{Exp}(1)$. Let $(\bar{U}_i)_i$ be an independent i.i.d. sequence with distribution $\bar{P}$. This ``marked'' Poisson process $(\bar{U}_i,T_i)_i$ supports a ``query operation'' given by the Poisson functional representation, where one can input a distribution $P$ over $\mathcal{U}$, and obtain one sample $\tilde{U}_P$ with distribution $P$. The Poisson functional representation is given by
\[
\tilde{U}_P := \bar{U}_K, \;\; \text{where}\; K:= \underset{i}{\arg \min} \; T_i \cdot \Big(\frac{\mathrm{d}P}{\mathrm{d}\bar{P}}(\bar{U}_i)\Big)^{-1}.
\]

The way this Poisson process is used in communication settings (e.g., in~\cite{li2021unified, liu2025one}) is that the encoder would query the process using the prior distribution of the signal to obtain the signal to be sent, and the decoder would query using the posterior distribution of the signal given the noisy observation to obtain the message. 
There is no error in the communication if the two queries return the same sample.
The probability of error can be bounded by the Poisson matching lemma~\cite{li2021unified}, shown as follows.

\begin{lemma}[Poisson matching lemma~\cite{li2021unified}]
Consider two distributions $P,Q \ll \bar{P}$. Almost surely, we have
\begin{equation*} 
    \mathbf{P}\big(
    \tilde{U}_Q\neq \tilde{U}_P\,\big|\, \tilde{U}_P\big) \leq 1 - \Big( 1 + \frac{\mathrm{d}P}{\mathrm{d}Q}(\tilde{U}_P)
    \Big)^{-1}. 
\end{equation*}
\end{lemma}

\medskip

\section{One-shot Information Hiding}
\label{sec::generalized_info_hiding}

In this section, we formulate the one-shot information hiding problem, provide our one-shot achievability results, show that they recover the existing asymptotic results, and discuss generalizations of the problem as well as some related problems.

\ifshortver
\begin{figure}[htpb]
    \centering
    \includegraphics[scale = 0.26]
    {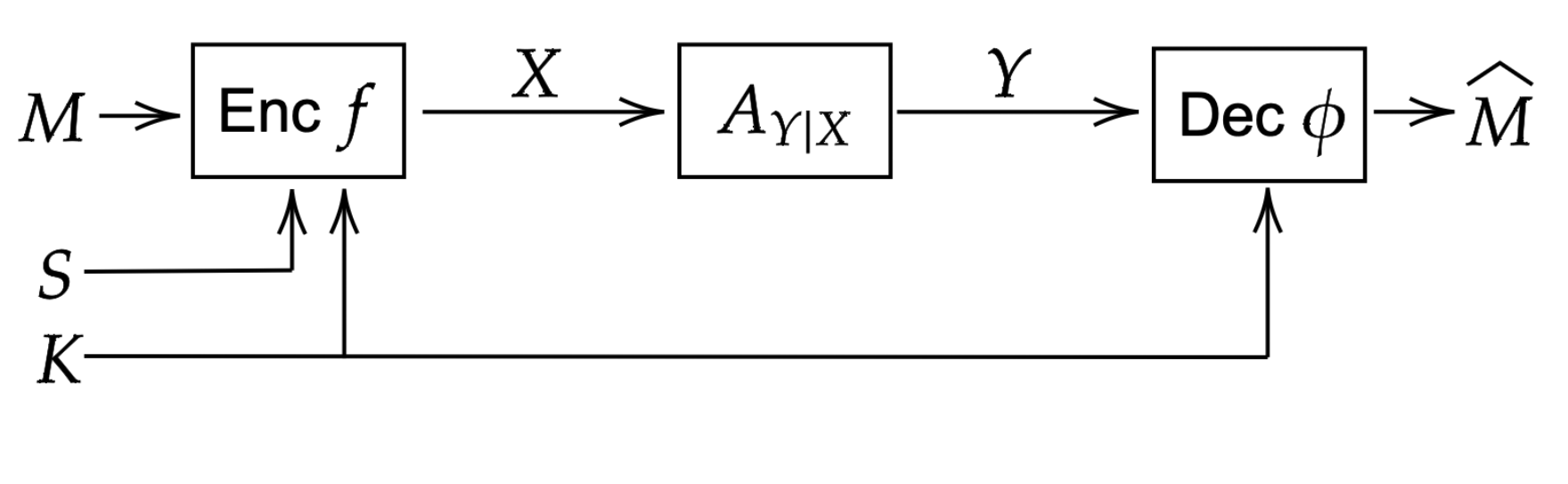}
    \caption{
    The information hiding problem. 
    }
    \label{fig:info_hiding}
\end{figure}
\else
\begin{figure*}[htpb]
    \centering
    \includegraphics[scale = 0.28]
    {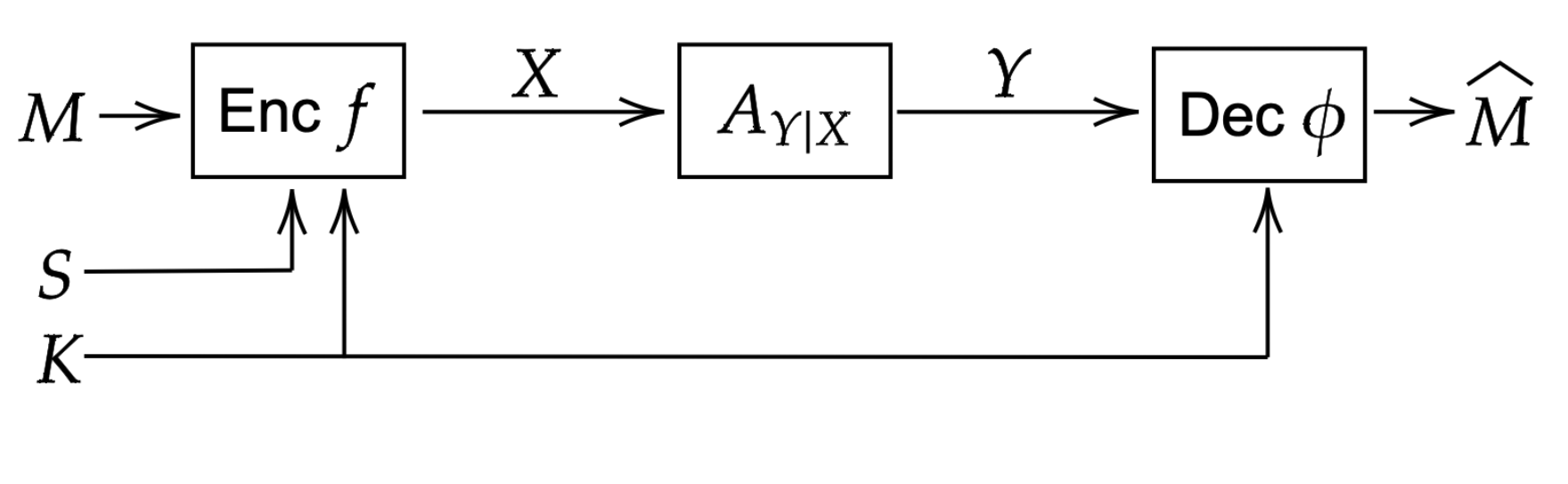}
    \caption{
    The information hiding problem. 
    }
    \label{fig:info_hiding}
\end{figure*}
\fi

The one-shot information hiding problem is described in Figure~\ref{fig:info_hiding}. 
A message $M$ is uniformly chosen from the set $[1:\mathsf{L}]$, where $\mathsf{L}$ is the \emph{message size}. 
The goal is to \emph{hide} the message $M$ into a host data source $S\in \mathcal{S}$ to produce $X$. 
The stegotext $X$ is sent through the attack channel $A_{Y|X}$, which is chosen by an attacker that attempts to remove or degrade the embedded message during the signal transmission. 
We allow the \emph{common randomness} $K\in \mathcal{K}$ to be available at both the encoder and the decoder, but not the attacker, and the decoder attempts to reconstruct the original message $M$ after observing $Y,K$. 
The common randomness $K$ reveals information about $S$ to the decoder, which may be correlated with $K$ according to the joint distribution $P_{S,K}$.

Given the random variables, the information hiding problem can be viewed as a \emph{game} between two parties: 
the first party consists of the encoder (information hider) and the decoder, who are cooperatively transmitting the message $M$; 
the second party is an attacker, who is trying to remove or degrade the hidden message $M$ in $S$ so that the decoder cannot correctly reconstruct it. 
See~\cite{moulin2003information, somekh2003error, somekh2004capacity, cohen2002gaussian} for more discussions on the game-theoretic formulation. 
Their roles and assumptions are elaborated as follows.

\begin{itemize}
    \item \textbf{Encoder}: 
    The encoder observes a message $M$ that is uniformly chosen from the set $[1:\mathsf{L}]$, and the goal of encoding is to \emph{hide} $M$ into a host data source $S \in \mathcal{S}$ by introducing some tolerable level of distortions. 
    Given $S,K,M$, the encoder outputs $X=f(S,K,M)$, where $f:\mathcal{S} \times \mathcal{K} \times [1:\mathsf{L}] \to \mathcal{X}$. 
    It is expected that $X$ is close to $S$, in the sense that $d_1(S,X)$ is small, where $d_1:\mathcal{S} \times \mathcal{X} \to [0,\infty)$ is a distortion measure. 
    We want $d_1(S,X) \leq \mathsf{D}_1$ with high probability. This will be elaborated later.
    The encoded steogtext $X$ is then transmitted through the a channel $A_{Y|X}\in \mathcal{A}$.

    \item \textbf{Attacker}: The attacker is formulated as a noisy channel $A_{Y|X}$, called the  \emph{attack channel}. 
    With input $X$, it performs data processing attacks by introducing another level of distortion and produces $Y$, a corrupted version of $X$. 
    Unlike the asymptotic study~\cite{moulin2003information}, we do not assume the attacker's strategy is known by the encoder and the decoder (See~\ref{subsec::disc} for discussions). 
    The attacker is free to choose $A_{Y|X}$ from a class of channels $\mathcal{A}$ (e.g., the class of channels satisfying some distortion constraint between $X$ and $Y$, or the class of memoryless channels where $X$ and $Y$ are sequences), whose cardinality can be infinite. 
    Both deterministic and randomized attacks can be performed. 
    We assume the attacker has knowledge of the distributions (but not the values) of $S,M,K$, and also knows the code that the encoder-decoder team uses.

    \item \textbf{Decoder}: Upon observing the attacker's output $Y$ and the side information $K$, the decoder wishes to recover the message $M$. It outputs $\hat{M} = \phi(K,Y)$, where $\phi: \mathcal{K}\times \mathcal{Y}\rightarrow [1:\mathsf{L}]$. 
    The decoder is uninformed of the attacker's strategy. 
    We require the following worst case failure probability to be small: 
    \begin{equation}
    P_e := \sup_{A_{Y|X} \in \mathcal{A}}\mathbf{P}\big( d_1(S,X) > \mathsf{D}_1 \;\; \mathrm{OR} \;\; M \neq \hat{M}\big), \label{eq:hiding_fail}
    \end{equation}
    where $(S,K,M)\sim P_{S,K} \times \mathrm{Unif}[1:\mathsf{L}]$, $X=f(S,K,\allowbreak  M)$, $Y|X \sim A_{Y|X}$ and $\hat{M} = \phi(K,Y)$ in the probability.\footnote{Note that \cite{moulin2003information} imposes a constraint on the expected distortion $\mathbf{E}[d_1(S,X)]$, which is reasonable in the context of \cite{moulin2003information} because the memoryless assumption and the law of large numbers ensure that the actual distortion is close to the expected distortion. 
    Since we are considering a one-shot setting where we only assume the attack channel is chosen from a set $\mathcal{A}$, if constraint need to be specified, it might be more reasonable to consider $d_1(S,X) > \mathsf{D}_1$ as a failure event and bound the probability of failure, i.e., the excess distortion probability instead, compared to expected distortion.
    }
\end{itemize}

\subsection{One-shot Achievability Results}
\label{subsec::one_shot_gene_infohid_bd}

We then provide one-shot achievability results of the information hiding problem. 
Since we let the encoder-decoder team account for all possible attack channels in a set $\mathcal{A}$, the achievability results have to suffer a penalty depending on the ``size'' of $\mathcal{A}$.  
Though the cardinality of $\mathcal{A}$ could be infinite, we can often find a finite subset $\tilde{\mathcal{A}}$ such that every attack channel $A \in \mathcal{A}$ is close enough to some $\tilde{A} \in \tilde{\mathcal{A}}$. 
We capture this notion of size by the $\epsilon$-covering number defined below (see similar covering arguments in~\cite{blackwell1959capacity, moulin2003information}).

\begin{definition}
\label{def::cov_num}
Given a set of channels $\mathcal{A}$ from $\mathcal{X}$ to $\mathcal{Y}$, its $\epsilon$-\emph{covering number} is defined as
\ifshortver
\begin{align*}
N_{\epsilon}(\mathcal{A}) := \min \Big\{
& |\tilde{\mathcal{A}}|:\, \tilde{\mathcal{A}}\subseteq \mathcal{A},\,\,  \sup_{A \in \mathcal{A}} \,\, \min_{\tilde{A} \in \tilde{\mathcal{A}}} \,\,  \sup_{x \in \mathcal{X}}\\
& \qquad 
\left\Vert A_{Y|X}(\cdot|x) - \tilde{A}_{Y|X}(\cdot|x) \right\Vert_{\mathrm{TV}} \leq \epsilon\Big\},
\end{align*}
\else
\begin{align*}
N_{\epsilon}(\mathcal{A}) := \min \Big\{|\tilde{\mathcal{A}}|:\, \tilde{\mathcal{A}}\subseteq \mathcal{A},\,\,  \sup_{A \in \mathcal{A}} \,\, \min_{\tilde{A} \in \tilde{\mathcal{A}}} \,\,  \sup_{x \in \mathcal{X}} 
\left\Vert A_{Y|X}(\cdot|x) - \tilde{A}_{Y|X}(\cdot|x) \right\Vert_{\mathrm{TV}} \leq \epsilon\Big\},
\end{align*}
\fi
where $\Vert A_{Y|X}(\cdot|x) - \tilde{A}_{Y|X}(\cdot|x) \Vert_{\mathrm{TV}} \in [0,1]$ denotes the total variation distance between $A_{Y|X}(\cdot|x)$ (the distribution of $Y$ if $X = x$, and $Y$ follows $A_{Y|X}$) and $\tilde{A}_{Y|X}(\cdot|x)$.
\end{definition}

We now present the main result, which is a one-shot achievability result with a bound on the error probability in terms of $N_{\epsilon}(\mathcal{A})$ and the information density terms. 

\begin{theorem}
\label{thm:info_hiding_cover}
    Fix any $P_{U,X|S,K}$
    and channel $\hat{A}_{Y|X}$. 
    For any $\epsilon \geq 0$, there exists an information hiding scheme satisfying
    \begin{align*}
     P_e & \leq N_{\epsilon}(\mathcal{A})   \sup_{A_{Y|X} \in \mathcal{A}}    \mathbf{E}_{Y|X \sim A_{Y|X}}\Bigg[ 1 -  \mathbf{1}\{d_1(S,X)\leq \mathsf{D}_1  \} 
 \\
& \qquad \qquad \qquad   \cdot \Big(
1 + \mathsf{L}
2^{ -\hat{\iota}(U;Y|K) + \iota(U;S|K)}
\Big)^{-1}
\Bigg] + \epsilon ,
\end{align*}
where we assume $(S,K,U,X,Y)\sim P_{S,K} P_{U,X|S,K} A_{Y|X}$ in the expectation, and $\hat{\iota}(U;Y|K)$ is the information density computed by the joint distribution $P_{S,K} P_{U,X|S,K} \hat{A}_{Y|X}$ (instead of $A_{Y|X}$), assuming that $\iota(U;S|K) , \hat{\iota}(U;Y|K)$ are almost surely finite for every $A_{Y|X} \in \mathcal{A}$.
\end{theorem}

\begin{proof}

The idea is that we design the decoder assuming that the attack channel is fixed to $\hat{A}_{Y|X}$, and prove that that this decoder works for every attack channel $A_{Y|X} \in \mathcal{A}$. 
Let $\mathcal{C} := ((\bar{U}_i, \bar{M}_i),T_i)_i$ where $(T_i)_i$ is a Poisson process, $\bar{U}_i \stackrel{\mathrm{iid}}{\sim} P_U$, and $\bar{M}_i \stackrel{\mathrm{iid}}{\sim} P_M$ (where $P_M = \mathrm{Unif}[1:\mathsf{L}]$). This will act as a random codebook shared between the encoder and the decoder (and this codebook will be fixed later).

The encoder observes the message $M\sim P_M$, the host signal $S$ and side information $K$, by the Poisson functional representation~\cite{li2018strong,li2021unified} on the distribution $P_{U|S,K}(\cdot | S, K)\times \delta_M$ over $\mathcal{U} \times [1:\mathsf{L}]$ it produces $U = \tilde{U}_{P_{U|S,K}(\cdot | S, K)\times \delta_M}$,\footnote{The Poisson functional representation produces a pair $(\tilde{U},\tilde{M})$, and $U$ is set to be the first component of the pair.} and sends the generated $X | (S,K,U) \sim P_{X|S,K,U}$.
The decoder observes $Y,K$ and outputs $\hat{M} = \tilde{M}_{\hat{P}_{U|Y,K}(\cdot | Y,K) \times P_M} $ by the Poisson functional representation, where $\hat{P}_{U|Y,K}$ is the conditional probability distribution computed by the joint distribution $P_{S,K} P_{U,X|S,K} \hat{A}_{Y|X}$. When the attack channel is $A_{Y|X} \in \mathcal{A}$, the error probability is bounded as follows: 
\begin{align*}
& P_e(A) := 1 - \mathbf{P}_{Y|X \sim A_{Y|X}}\big( d_1(S,X)\leq \mathsf{D}_1  \;\;  \mathrm{ AND } \;\; M= \hat{M}  \big)\\
& = \mathbf{E}
\bigg[ 1 - \mathbf{1}\{d_1(S,X)\leq \mathsf{D}_1  \}  \cdot \mathbf{1}\{M = \hat{M} \} 
\bigg] \\
& =\mathbf{E}
\bigg[ 1 - \mathbf{1}\{d_1(S,X)\leq \mathsf{D}_1  \} \cdot \mathbf{P}\big(M = \hat{M} | M,S,U,Y,K \big) \bigg]  \\
& \leq\mathbf{E}
\bigg[ 1 - \mathbf{1}\{d_1(S,X)\leq \mathsf{D}_1  \}  \\
& \qquad \cdot 
\mathbf{P}\big((U,M) = (\tilde{U}, \tilde{M})_{\hat{P}_{U|Y,K}(\cdot|Y,K)\times P_M} | M,S,U,Y,K\big) \bigg]  \\
& \stackrel{(a)}{\leq} \mathbf{E}
\bigg[ 1 - \mathbf{1}\{d_1(S,X)\leq \mathsf{D}_1  \}   \\
&\qquad \cdot
\Big(
1 +  \frac{\mathrm{d} P_{U|S,K} (\cdot | S,K)\times \delta_M}{\mathrm{d}  \hat{P}_{U|Y,K}(\cdot | Y, K)\times P_M }
(U, M)
\Big)^{-1} \bigg]  \\
& =
\mathbf{E}\bigg[
1 -  \mathbf{1}\{d_1(S,X) \leq  \mathsf{D}_1  \}  \Big(
1 + \mathsf{L}
2^{ -\hat{\iota}(U;Y|K) + \iota(U;S|K)}
\Big)^{  -1}
\bigg] \\
& \leq 
\sup_{A_{Y|X} \in \mathcal{A}}\mathbf{E}_{Y|X \sim A_{Y|X}}\bigg[
1 -  \mathbf{1}\{d_1(S,X)\leq \mathsf{D}_1  \} \\
& \qquad   \cdot \Big(
1 + \mathsf{L}
2^{ -\hat{\iota}(U;Y|K) + \iota(U;S|K)}
\Big)^{-1}
\bigg] \quad =: \overline{P_e},
\end{align*}
where $(a)$ is by the Poisson matching lemma.\footnote{The Poisson matching lemma is applied on the conditional distributions given $M,S,U,Y,K$. Also see the conditional Poisson matching lemma~\cite{li2021unified}.} 
If we allow the encoder and the decoder to share unlimited additional common randomness, we can assume the codebook $\mathcal{C} = ((\bar{U}_i, \bar{M}_i),T_i)_i$ is actually shared, and conclude that $P_e = \sup_{A \in \mathcal{A}}P_e(A) \leq \overline{P_e}$. Nevertheless, the only actual common randomness between the encoder and the decoder is $K$, which we cannot control. Therefore, we have to fix the codebook.

Let $P_e(A,c)$ be the probability of error when the attack channel is $A$ and the codebook is $\mathcal{C}=c$. We have $P_e(A) = \mathbf{E}_{\mathcal{C}}[P_e(A,\mathcal{C})]$. 
Let $\tilde{\mathcal{A}}\subseteq \mathcal{A}$ attain the minimum in $N_{\epsilon}(\mathcal{A})$.

Consider any $A \in \mathcal{A}$, and let $\tilde{A} \in \tilde{\mathcal{A}}$ satisfy 
\begin{equation*}
    \sup_{x \in \mathcal{X}}
    \left \Vert A_{Y|X} (\cdot|x) - \tilde{A}_{Y|X}(\cdot|x) 
    \right \Vert_{\mathrm{TV}} \leq \epsilon.
\end{equation*} 
The total variation distance between the joint distribution of $M,S,K,U,X,Y$ under the attack channel $A$ conditional on $\mathcal{C}=c$ and the joint distribution under the attack channel $\tilde{A}$ conditional on $\mathcal{C}=c$ is also bounded by $\epsilon$. Hence $|P_e(A,c) - P_e(\tilde{A},c)| \leq \epsilon$ and
\begin{align*}
 P_e(A,c) &\leq P_e(\tilde{A},c)+\epsilon \\
&\leq \sum_{\tilde{A}\in \tilde{\mathcal{A}}}P_e(\tilde{A},c)+\epsilon.
\end{align*}  
Therefore,
\begin{align*}
 \mathbf{E}_{\mathcal{C}}\Big[\sup_{A\in \mathcal{A}} P_e(A,\mathcal{C})\Big] &  \leq \mathbf{E}_{\mathcal{C}}\Big[\sum_{\tilde{A}\in \tilde{\mathcal{A}}}P_e(\tilde{A},\mathcal{C})+\epsilon\Big]  \\
& = \sum_{\tilde{A}\in \tilde{\mathcal{A}}} P_e(\tilde{A})+\epsilon \\
& \leq |\tilde{\mathcal{A}}| \cdot \overline{P_e} +\epsilon.
\end{align*} 
The proof is completed by the existence of a codebook $c$ such that $\sup_{A\in \mathcal{A}} P_e(A,c) \leq |\tilde{\mathcal{A}}| \cdot \overline{P_e} +\epsilon$.
\end{proof}

\medskip

\begin{remark}
\label{remark::recover_GP}
    Note that when $K=\emptyset$, $d_1(s,x)=0$, and $\mathcal{A}=\{A_{Y|X}\}$ is a singleton set, taking $\hat{A}_{Y|X}=A_{Y|X}$, Theorem~\ref{thm:info_hiding_cover} recovers the one-shot Gelfand-Pinsker coding result in~\cite{li2021unified}, which is the only known one-shot bound that attains  the best known second order result in~\cite{scarlett2015dispersions}. 
    The asymptotic Gelfand-Pinsker capacity~\cite{gelfand1980coding,heegard1983capacity} can be readily recovered. 
\end{remark}

\begin{remark}
The one-shot achievability results can be converted to a finite blocklength result where $n$ is a fixed number using the Berry-Esseen theorem~\cite{berry1941accuracy, esseen1942liapunov, feller1971introduction}. 
For example, see~\cite[Section IV]{li2021unified} for converting the one-shot result of the Gelfand-Pinsker problem to a second-order result by the Berry-Esseen theorem~\cite{berry1941accuracy, esseen1942liapunov, feller1971introduction}, which can also be recovered by Theorem~\ref{thm:info_hiding_cover}. 
\end{remark}

\subsection{Recovery of the Asymptotic Information Hiding Capacity}
\label{sec::recovery_hiding}

In this subsection, we show that our Theorem~\ref{thm:info_hiding_cover} recovers the asymptotic information hiding capacity~\cite{moulin2003information} when applied to the discrete and memoryless setting, with the attack channels subject to a distortion constraint. 
This gives an alternative proof to~\cite{moulin2003information}.
Note that, from a similar procedure, we can recover the achievable bound for information hiding with stegotext reconstruction~\cite[Theorem 1]{xu2023information} (which, in turn, is an extension of~\cite{grover2015information} and~\cite{sumszyk2009information}) by using Theorem~\ref{thm:info_hiding_cover2}. 
We choose to present the details of Theorem~\ref{thm:info_hiding_cover} recovering the information hiding capacity~\cite{moulin2003information} for simplicity, as it already captures the core idea.

We first provide a simple bound on the $\epsilon$-covering number in the case that $X$ and $Y$ are discrete and finite.

\begin{proposition}
\label{prop:covering_bound}

If $\mathcal{X}$ and $\mathcal{Y}$ are finite, then 
\[
N_{\epsilon}(\mathcal{A})\leq \Big(\frac{1}{2\epsilon} + \frac{|\mathcal{Y}|+1}{2}\Big)^{|\mathcal{X}|\cdot|\mathcal{Y}|}.
\]
\end{proposition}

\begin{proof}
Write 
\begin{equation*}
  d(A,\tilde{A}):= \sup_{x \in \mathcal{X}}\Vert A_{Y|X}(\cdot|x) - \tilde{A}_{Y|X}(\cdot|x) \Vert_{\mathrm{TV}}.   
\end{equation*}
We use the standard method to bound the covering number, where we
start with $\tilde{\mathcal{A}}=\emptyset$, and add $A\in\mathcal{A}$
not currently covered by $\tilde{\mathcal{A}}$ (i.e., $\min_{\tilde{A}\in\tilde{\mathcal{A}}}d(A,\tilde{A})>\epsilon$)
to $\tilde{\mathcal{A}}$ one by one until all of $\mathcal{A}$ is
covered. Note that every two different $\tilde{A},\tilde{A}'\in\tilde{\mathcal{A}}$
produced this way must satisfy $d(\tilde{A},\tilde{A}')>\epsilon$,
and hence the $(\epsilon/2)$-balls $\{A:\,d(A,\tilde{A})\le\epsilon/2\}$
must be disjoint for $\tilde{A}\in\tilde{\mathcal{A}}$.

We now treat $A_{Y|X}$ as a transition probability matrix $A \in \mathbb{R}^{|\mathcal{Y}|\times|\mathcal{X}|}$. We have \begin{align*}
    d(A,\tilde{A}) &= \frac{1}{2}\Vert A-\tilde{A}\Vert_1 \\
     &= \frac{1}{2} \max_x \sum_y |A_{y,x} - \tilde{A}_{y,x}|. 
\end{align*}
The volume of the ball $\{A \in \mathbb{R}^{|\mathcal{Y}|\times|\mathcal{X}|}: \,d(A,\tilde{A})\le\epsilon/2\}$ (i.e., its Lebesgue measure in the space $\mathbb{R}^{|\mathcal{Y}|\cdot |\mathcal{X}|}$)
is $((2\epsilon)^{|\mathcal{Y}|}/(|\mathcal{Y}|!))^{|\mathcal{X}|}$, and all these balls are subsets of $\{A \in \mathbb{R}^{|\mathcal{Y}|\times|\mathcal{X}|}: \min_{x,y}A_{y,x}\ge-\epsilon,\,\max_x \sum_{y}A_{y,x}\le1+\epsilon\}$,
which has a volume $((1+(|\mathcal{Y}|+1)\epsilon)^{|\mathcal{Y}|}/(|\mathcal{Y}|!))^{|\mathcal{X}|}$. Hence, the size of $\tilde{\mathcal{A}}$ is upper-bounded by
\begin{align*}
&\frac{\Big(\big(1+(|\mathcal{Y}|+1)\epsilon \big)^{|\mathcal{Y}|}/(|\mathcal{Y}|!)\Big)^{|\mathcal{X}|}}{\big((2\epsilon)^{|\mathcal{Y}|}/(|\mathcal{Y}|!)\big)^{|\mathcal{X}|}} 
=\Big(\frac{1}{2\epsilon} + \frac{|\mathcal{Y}|+1}{2}\Big)^{|\mathcal{X}|\cdot|\mathcal{Y}|}.
\end{align*}
\end{proof}

We now show that Theorem~\ref{thm:info_hiding_cover} recovers the asymptotic result in \cite{moulin2003information} when $S,K,X,Y$ are finite and discrete, and the attack channel is memoryless and is subject to a distortion constraint.

Consider sequences $S^n=(S_1,\ldots,S_n)$, $K^n$, $X^n$, $Y^n$ where $(S_i,K_i)\stackrel{\mathrm{iid}}{\sim} P_{S,K}$. Consider a channel input distribution $P_X$. The class of attack channels $\mathcal{A}_n=\mathcal{A}_n(P_X)$ (which depends on $P_X$) is taken to be 
\begin{equation*}
    \mathcal{A}_n(P_X) := \big\{A_{Y|X}^n:\, A_{Y|X} \in \mathcal{A}(P_X) \big\},
\end{equation*}
and we let 
\begin{equation*}
    \mathcal{A}(P_X) := \big\{A_{Y|X}:\,  \mathbf{E}_{(X,Y) \sim P_X A_{Y|X}}[d_2(X,Y)] \leq \mathsf{D}_2\big\},
\end{equation*}
where $d_2: \mathcal{X} \times \mathcal{Y} \to [0,\infty)$ is a distortion measure, and $\mathsf{D}_2$ is the allowed distortion level. In other words, the attacker can only use memoryless channels $A_{Y|X}^n$ that satisfy the expected distortion constraint $\mathbf{E}[d_2(X,Y)] \leq \mathsf{D}_2$.
The asymptotic hiding capacity given in~\cite{moulin2003information} is
\begin{equation*}
	C = \max_{P_{U,X|S,K}} \min_{A_{Y|X}:\, \mathbf{E}[d_2(X,Y)] \leq \mathsf{D}_2} \big( I(U;Y|K) - I(U;S|K) \big).
\end{equation*}
where the maximum is over $P_{U,X|S,K}$ with $\mathbf{E}[d_1(S,X)]\le\mathsf{D}_{1}$.

We now show the achievability of the above asymptotic rate as a direct corollary of Theorem~\ref{thm:info_hiding_cover}. Fix $P_{U,X|S,K}$ which achieves the above maximum 
subject to $\mathbf{E}[d_1(S,X)]\le\mathsf{D}'_{1}$ where $\tilde{\mathsf{D}}_1 < \mathsf{D}_1$.
Take $\hat{A}_{Y|X}$ to be the minimizer of the rate-distortion function
$\min_{A_{Y|X}:\, \mathbf{E}[d_2(X,Y)] \leq \mathsf{D}_2} I(U;Y|K)$,
and assume $(S,K,U,X,Y)\sim P_{S,K} P_{U,X|S,K} \hat{A}_{Y|X}$.
Write the information density and mutual information obtained from this distribution as $\hat{\iota}_{U;Y|K}$ and $\hat{I}(U;Y|K)$, respectively. Fix a coding rate $R < \hat{I}(U;Y|K)- I(U;S|K)$. We want to show that this rate is achievable.

Consider any attack channel $A_{Y|X}$ with $\mathbf{E}[d_2(X,Y)] \leq \mathsf{D}_2$. 
Let $A_{Y|X}^{\lambda}:=(1-\lambda)\hat{A}_{Y|X}+\lambda A_{Y|X}$ for $0\leq \lambda\leq 1$. Write $I_{\lambda}(U;Y|K)$ for the mutual information computed assuming $Y|X \sim A_{Y|X}^{\lambda}$. It is straightforward to check that
\[
\frac{\mathrm{d}}{\mathrm{d}\lambda}I_{\lambda}(U;Y|K)\Big|_{\lambda=0}  =\mathbf{E}_{Y|X\sim A_{Y|X}}[\hat{\iota}(U;Y|K)]-\hat{I}(U;Y|K).
\]
By the optimality of $\hat{A}$, the above derivative is nonnegative, and hence $\mathbf{E}_{Y|X\sim A_{Y|X}}[\hat{\iota}(U;Y|K)] \geq \hat{I}(U;Y|K)$. Therefore, when we have i.i.d. sequences $(S^n,K^n,U^n,X^n,Y^n)\sim P^n_{S,K} P^n_{U,X|S,K} A^n_{Y|X}$ and $\mathsf{L}=\lfloor 2^{nR}\rfloor$, by the law of large numbers, 
\begin{align*}
& \mathsf{L} 2^{-\hat{\iota}(U^n;Y^n|K^n)+\iota(U^n;S^n|K^n)} \\
& \leq 2^{nR-\sum_{i=1}^n (\hat{\iota}(U_i;Y_i|K_i)-\iota(U_i;S_i|K_i))} \\
& \to \; 0
\end{align*}
exponentially as $n\to \infty$ since $\mathbf{E}[\hat{\iota}(U_i;Y_i|K_i)-\iota(U_i;S_i|K_i))] \geq \hat{I}(U;Y|K)-I(U;S|K) > R$. We also have $d_1(S^n,X^n) = \sum_{i=1}^n d_1(S_i,X_i) > n\mathsf{D}_1$ with probability approaching $0$ exponentially since $\tilde{\mathsf{D}}_1 < \mathsf{D}_1$. These convergences are uniform over all such attack channels $A_{Y|X}$ since the random variables are discrete and finite. 

Therefore, to bound $P_e$ using Theorem~\ref{thm:info_hiding_cover}, it is left to bound the $\epsilon$-covering number $N_{\epsilon}(\mathcal{A}_n(P_X))$. Note that 
\begin{align*}
    & \left \Vert A_{Y|X}^n(\cdot | x^n) - \tilde{A}_{Y|X}^n(\cdot | x^n) 
    \right \Vert_{\mathrm{TV}} \\
    & \leq \sum_{i=1}^n 
    \left \Vert A_{Y|X}(\cdot | x_i)-\tilde{A}_{Y|X}(\cdot | x_i) 
    \right \Vert_{\mathrm{TV}}, 
\end{align*}
and hence we can construct an $\epsilon$-cover of $\mathcal{A}_n(P_X)$ using an $(\epsilon/n)$-cover of $\mathcal{A}(P_X)$. Therefore, \begin{align*}
    & N_{\epsilon}(\mathcal{A}_n(P_X)) \\
    & \leq N_{\epsilon/n}(\mathcal{A}(P_X)) \\
    & = O((n/\epsilon)^{|\mathcal{X}|\cdot|\mathcal{Y}|})
\end{align*}
by Proposition~\ref{prop:covering_bound}, which grows much slower than the exponential decrease of the expectation in Theorem~\ref{thm:info_hiding_cover}. Therefore, taking $\epsilon =1/n$, we have $P_e \to 0$ as $n\to \infty$. Taking $\tilde{\mathsf{D}}_1 \to \mathsf{D}_1$ completes the proof.

\subsection{Discussions} 
\label{subsec::disc}

In~\cite{moulin2003information}, it is assumed that the attack channel must be memoryless, and the decoder can obtain full knowledge of the attack channel. 
This assumption can be reasonably justified by the large blocklength, as one can use training symbols at the beginning of transmission, with their size becoming negligible compared to the blocklength.
However, this assumption becomes questionable in the one-shot scenario, where the attack channel can be arbitrary and is only used once. 
We drop this assumption and allow the decoder to be totally uninformed of the attack channel.
In~\cite{somekh2004capacity} this assumption is also dropped, and an asymptotic hiding capacity expressed as the limit of a sequence of single-letter expressions has been derived using constant composition codes. The key difference between \cite{somekh2004capacity} and our setting is that the side information in \cite{somekh2004capacity} is a shared key of unlimited size independent of $M, S$ that can be chosen as a part of the coding scheme, whereas in our paper and \cite{moulin2003information} the $K$ is a given side information that may be correlated with $S$ (where the dependence is from the joint distribution $P_{S,K}$), and cannot be changed. 
In some watermarking problems~\cite{cox1997secure, hartung1999multimedia} certain components can be further constrained, e.g., there may exist a mapping from the message $M$ to a codeword $V(M)$ which is independent of $S$, and then composite data are obtained by a mapping from $S$, $K$ and $V(M)$.

The information hiding setting can be regarded as a variant of Gelfand-Pinsker coding for channels with side information at the encoder~\cite{gelfand1980coding,heegard1983capacity}, where the channel is fixed and not chosen by the attacker, and there is no shared side information between the encoder and the decoder. 
Since the encoder and the decoder have to account for all possible attack channels, this can be regarded as a combination of Gelfand-Pinsker coding and compound channel~\cite{blackwell1959capacity, dobrushin1959optimum, wolfowitz1980simultaneous}. 
The analyses in~\cite{moulin2003information, somekh2004capacity} utilize techniques such as random binning, joint typicality decoding and constant composition codes, which are also commonly utilized in the asymptotic analyses of Gelfand-Pinsker coding~\cite{gelfand1980coding,scarlett2015dispersions}.
These techniques may not be suitable for our one-shot setting. 
Strong typicality and constant composition codes are inapplicable when the blocklength is $1$. While random binning can be applied to one-shot Gelfand-Pinsker coding~\cite{verdu2012non,yassaee2013technique,watanabe2015nonasymp}, it produces weaker results compared to the Poisson matching lemma~\cite{li2021unified}. 
To obtain tight one-shot bounds for information hiding, we utilize the Poisson matching lemma instead, which has been shown to perform well in various one-shot settings~\cite{li2021unified, liu2025one}.

\subsection{Generalizations}
\label{subsec::hiding_generalization}

We discuss some generalizations of the information hiding setting. 
As shown in Figure~\ref{fig:info_hiding}, we have a host signal $S$ available to the encoder and a side information $K$ available to both the encoder and the decoder. 
This can be generalized by letting the encoder have a side information source $S_{\mathrm{Enc}}$, the decoder have another side information source $S_{\mathrm{Dec}}$, and they possibly share a codebook $\mathcal{C}$. 
In this case, our setting (also~\cite{moulin2003information} and public watermarking~\cite{somekh2004capacity}) correspond to $S_{\mathrm{Enc}} = (S, K)$ and $S_{\mathrm{Dec}} = K$, while private watermarking~\cite{somekh2003error} corresponds to $S_{\mathrm{Enc}} = S_{\mathrm{Dec}} = S$. 
This setting has been investigated in~\cite{moulin2007capacity}. 
Moreover, in~\cite{grover2015information, xu2023information}, the decoder not only recovers the message $M$, but also lossily reconstructs the stegotext $X$, within another level of distortion $d_1'(X, \hat{X})$. 
See Figure~\ref{fig:generalized_info_hiding} for an illustration.

\ifshortver
\begin{figure}[htpb]
    \centering
    \includegraphics[scale = 0.3]
    {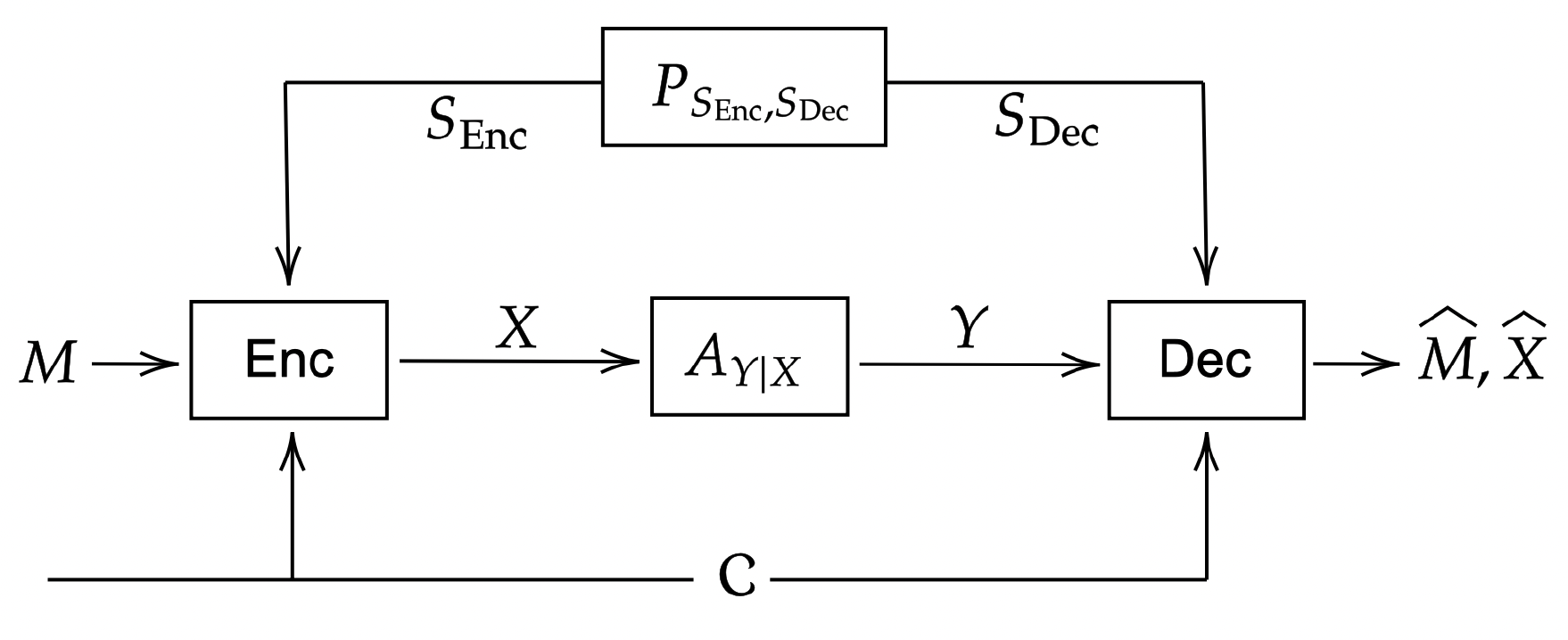}
    \caption{
    The generalized information hiding problem. 
    }
    \label{fig:generalized_info_hiding}
\end{figure}   
\else
\begin{figure*}[htpb]
    \centering
    \includegraphics[scale = 0.3]
    {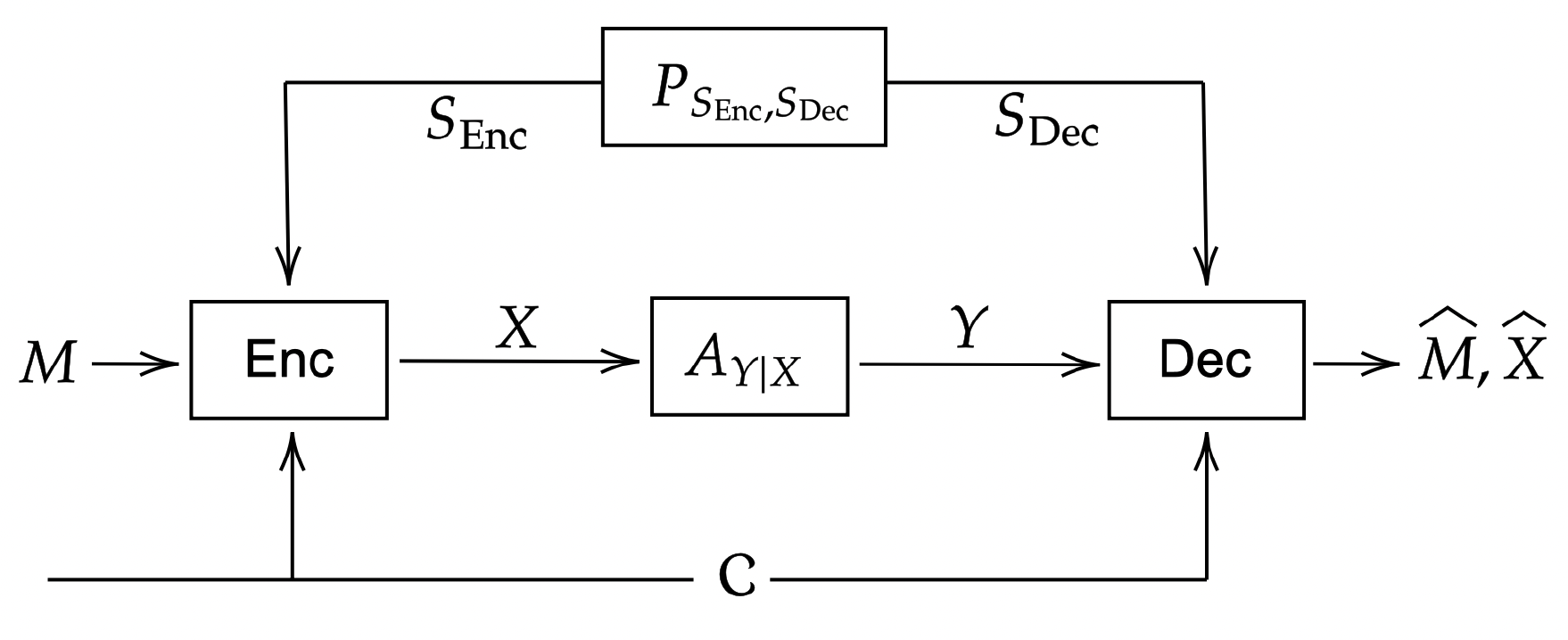}
    \caption{
    The generalized information hiding problem. 
    }
    \label{fig:generalized_info_hiding}
\end{figure*}   
\fi

By employing these generalizations, the information hiding setting can be extended as follows: Upon observing side information $S_{\mathrm{Enc}}$ and the message $M$, the encoding function $f: \mathcal{S}_{\mathrm{Enc}} \times [1:\mathsf{L}] \to \mathcal{X}$ outputs $X = f(S_{\mathrm{Enc}}, M)$; the decoder observes $Y$, the output of the attacker, together with another side information $S_{\mathrm{Dec}}$, and computes $(\hat{M}, \hat{X}) = \phi(S_{\mathrm{Dec}}, Y)$, the distorted versions of $M$ and $X$, where $\phi: \mathcal{S}_{\mathrm{Dec}} \times \mathcal{Y} \to [1:\mathsf{L}] \times \mathcal{X}$. 
To highlight the potential advantage of \emph{randomized} coding schemes, we allow both the encoding function $f$ and decoding function $g$ to depend on a random variable $\mathcal{C}$, which is shared between the encoder and decoder but unknown to the attacker. This random variable $\mathcal{C}$ serves as a source of common randomness (or equivalently, as a shared codebook).
The distortion of the stegotext $X$ is measured by $d_1'(X, \hat{X})$, where $d_1': \mathcal{X} \times \mathcal{X} \to [0, \infty)$ is another distortion measure. 
The decoder is uninformed of the attack channel, and we intend to bound the worst case failure probability
\ifshortver
\begin{align}
P_e := \sup_{A_{Y|X} \in \mathcal{A}}\mathbf{P}\big( 
& d_1(S,X) > \mathsf{D}_1 \;\; \mathrm{OR} \nonumber \\
& d_1'(X, \hat{X}) > \mathsf{D}_1' \;\; \mathrm{OR} \;\; M \neq \hat{M}\big), \label{eq:hiding_fail}
\end{align}
\else
\begin{equation}
P_e := \sup_{A_{Y|X} \in \mathcal{A}}\mathbf{P}\big( d_1(S,X) > \mathsf{D}_1 \;\; \mathrm{OR} \;\; d_1'(X, \hat{X}) > \mathsf{D}_1' \;\; \mathrm{OR} \;\; M \neq \hat{M}\big), \label{eq:hiding_fail}
\end{equation}
\fi
where we assume $(S_{\mathrm{Enc}},S_{\mathrm{Dec}},M)\sim P_{S_{\mathrm{Enc}},S_{\mathrm{Dec}}} \times \mathrm{Unif}[1:\mathsf{L}]$, $X=f(S_{\mathrm{Enc}},M)$, $Y|X \sim A_{Y|X}$ and $(\hat{M}, \hat{X}) = \phi(S_{\mathrm{Dec}}, Y)$ in the probability. 

Based on this setting, we can derive the following theorem, which is a generalized version of Theorem~\ref{thm:info_hiding_cover}. 
The proof is similar to that of Theorem~\ref{thm:info_hiding_cover} and is therefore omitted.

\begin{theorem}
\label{thm:info_hiding_cover2}
    Fix any $P_{U,X|S_{\mathrm{Enc}},S_{\mathrm{Dec}}}$
    and channel $\hat{A}_{Y|X}$. 
    For any $\epsilon \geq 0$, there exists a scheme for the information hiding problem satisfying
    \ifshortver
    \begin{align*}
    & P_e  \leq N_{\epsilon}(\mathcal{A})\,\cdot\, \sup_{A \in \mathcal{A}} \, \mathbf{E}_{Y|X \sim A}\Bigg[ 1 -  \mathbf{1}\{d_1(S_{\mathrm{Enc}}, X)\leq \mathsf{D}_1  \} 
     \\
    & \cdot \mathbf{1}\{d_1'(X, \hat{X})\leq \mathsf{D}_1'  \} \cdot \Big(
    1 + \mathsf{L} \cdot
    2^{ -\hat{\iota}(U; Y, S_{\mathrm{Dec}}) + \iota(U;S_{\mathrm{Enc}})}
    \Big)^{-1}
    \Bigg]  \\
    & \qquad \qquad \qquad \qquad \qquad \qquad \qquad \qquad \qquad \qquad \qquad + \epsilon,  
\end{align*}
\else
    \begin{align*}
    & P_e  \leq N_{\epsilon}(\mathcal{A})\,\cdot\, \sup_{A \in \mathcal{A}} \, \mathbf{E}_{Y|X \sim A}\Bigg[ 1 -  \mathbf{1}\{d_1(S_{\mathrm{Enc}}, X)\leq \mathsf{D}_1  \} 
     \\
    & \qquad  \qquad \cdot \mathbf{1}\{d_1'(X, \hat{X})\leq \mathsf{D}_1'  \} \cdot \Big(
    1 + \mathsf{L} \cdot
    2^{ -\hat{\iota}(U; Y, S_{\mathrm{Dec}}) + \iota(U;S_{\mathrm{Enc}})}
    \Big)^{-1}
    \Bigg] + \epsilon,  
\end{align*}
\fi
where we assume $(S_{\mathrm{Enc}} ,S_{\mathrm{Dec}},U,X,Y)\sim P_{S_{\mathrm{Enc}} ,S_{\mathrm{Dec}}} P_{U,X|S_{\mathrm{Enc}}} A_{Y|X}$ in the expectation, and $\hat{\iota}(U; Y, S_{\mathrm{Dec}})$ is the information density computed by the joint distribution $P_{S,K} P_{U,X|S,K} \hat{A}_{Y|X }$ (instead of $A_{Y|X}$), assuming that $\iota(U;S_{\mathrm{Enc}}), \hat{\iota}(U; Y, S_{\mathrm{Dec}})$ are almost surely finite for every $A_{Y|X} \in \mathcal{A}$.
\end{theorem}

\section{One-shot Compound Wiretap Channels}
\label{sec::wiretap}

In this section, we consider the compound wiretap channel~\cite{liang2009compound} in the one-shot setting. 
We utilize the Poisson matching lemma~\cite{li2021unified}, under a framework similar to the one-shot information hiding discussed in Section~\ref{sec::generalized_info_hiding}. 
We provide novel one-shot achievablity results for the compound wiretap channel~\cite{liang2009compound}. 
To the best of our knowledge, the one-shot results of this problem has not been discussed in literature, though finite-blocklength bounds on the single wiretap channel (without uncertainties/multiple-states of the channel) can be found in~\cite{hayashi2006general, liu2016e_, yang2019wiretap}. 

Unlike the asymptotic analysis of the discrete memoryless compound wiretap channel~\cite{liang2009compound}, our one-shot results also apply to \emph{continuous} scenarios. 
Note that~\cite{schaefer2015secrecy} also studied the continuous case of compound wiretap channels, but the focus in~\cite{schaefer2015secrecy} was mainly on the compound Gaussian MIMO wiretap channels, and the analysis was not one-shot.

\subsection{Problem Formulation}
\label{subsec::comp_wiretap}

The one-shot compound wiretap channel setting is described as follows. 
A message $M$ is uniformly chosen from $\mathrm{Unif}[\mathsf{L}]$. 
Upon observing $M\sim \mathrm{Unif}[\mathsf{L}]$, the encoder produces $X = f(M)$, where $f:[\mathsf{L}]\rightarrow \mathcal{X}$ is a randomized encoding function. 
Then $X$ is sent through a channel $P_{Y,Z|X}$ that is unknown to the encoder and the decoder but known to the eavesdropper. 
A legitimate decoder observes $Y$ and recovers $\hat{M} = g(Y)$, where $g:\mathcal{Y}\rightarrow[\mathsf{L}]$ is a decoding function. 
The eavesdropper observes $Z\in \mathcal{Z}$. 
Justified by~\cite{liang2009compound} and~\cite[Lemma 1]{liang2008multiple}, we can assume the transition probability distribution is $P_{Y | X} P_{Z | X}$ by decomposing $P_{Y,Z|X}$ without loss of optimality.

In this paper, we do not fix $P_{Y|X}$, but rather allow $P_{Y | X}$ to be any element in $\mathcal{D}$, a set of channels to the legitimate decoder. 
Similarly, we allow $P_{Z | X}$ to be any element in $\mathcal{E}$, a set of channels to the eavesdropper. 
We assume the encoder and decoder have no knowledge of $P_{Y|X}$ and $P_{Z|X}$, making them \emph{universal} in the sense that their performance guarantees do not depend on specific channel realizations. 
For the secrecy guarantee, we make the conservative assumption that the eavesdropper knows both $P_{Y|X}$ and $P_{Z|X}$ that are in use.

Unlike~\cite{liang2009compound} but similar to our discussions on the information hiding problem, the set of channels to the legitimate decoder $\mathcal{D}$ and the set of channels to the eavesdropper $\mathcal{E}$ can be infinite, which captures the infinite variability of real-world signals and their propagation characteristics in practical applications. 
While the cardinalities of $\mathcal{D}, \mathcal{E}$ can be infinite, we can often find finite subsets $\tilde{\mathcal{D}}$, $\tilde{\mathcal{E}}$ such that every channel to the decoder (or eavesdropper) in $\tilde{\mathcal{D}}$ (or $\tilde{\mathcal{E}}$) would be close enough to some $\tilde{P}_{Y|X} \in \tilde{\mathcal{D}}$ (or $\tilde{P}_{Z|X} \in \tilde{\mathcal{E}}$). 
This idea has appeared in Section~\ref{sec::generalized_info_hiding} and also in~\cite{schaefer2015secrecy, blackwell1959capacity}.

The objective is to bound the worst case error probability
\begin{equation}
    P_e := \sup_{P_{Y|X}\in \mathcal{D}} \mathbf{P}\big(
    M\neq \hat{M}
    \big),
\end{equation}
while also ensure the secrecy is guaranteed, which is measured by the total variation distance
\begin{equation}
     \gamma := \sup_{P_{Z|X}\in \mathcal{E}} \left\Vert P_{M,Z} - P_M\times P_Z\right\Vert_{\mathrm{TV}}
\end{equation}
being small. 

\subsection{One-Shot Achievability}
\label{subsec::comp_wiretap}

We then provide one-shot achievability results of the compound wiretap channel. 
The result can be viewed as a combination of the covering argument appeared in Section~\ref{sec::generalized_info_hiding} and the one-shot soft covering lemma in~\cite[Proposition 3]{li2021unified}. 
Other existing one-shot wiretap channel results~\cite{hayashi2006general, liu2016e_} might also be utilized in a similar framework as well.

\begin{theorem}
\label{thm:compound_wiretap}
     Fix any $P_{U,X}$ and any wiretap channel $\hat{P}_{Y,Z|X} = \hat{P}_{Y|X} P_{Z|X}$. 
    For any $\nu \geq 0$, any $\epsilon_1, \epsilon_2 \geq 0$ and $\mathsf{A}, \mathsf{B}\in\mathbb{N}$, there exists a code for the compound wiretap channel setting, with message $M\sim \mathrm{Unif}[\mathsf{L}]$, satisfying 
    \ifshortver
    \begin{align}
        & P_e +\nu\cdot \gamma \nonumber \\
        & \leq N_{\epsilon_1}(\mathcal{D}) 
        \sup_{P_{Y|X}\in \mathcal{D}} 
        \mathbf{E}_{Y|X\sim P_{Y|X}} 
        \Big[
        \min\left\{
        \mathsf{L}\mathsf{A}2^{- \hat{\iota}(U;Y)}, 1
        \right\}
        \Big] + \epsilon_1  \nonumber \\
        & \quad + \nu\cdot N_{\epsilon_2}(\mathcal{E}) 
        \Bigg(
        \sup_{\tilde{P}_{Z|X}\in \mathcal{E}} 
        2 \cdot \mathbf{E}_{Z|X\sim \tilde{P}_{Z|X}} 
        \Big[ 
        \big(1 + 2^{-\iota(U;Z)} \big)^{-\mathsf{B}} 
        \Big]  \nonumber \\
        & \qquad \qquad \qquad + \sqrt{ \mathsf{B}\mathsf{A}^{-1}}
        \Bigg)
        +  \nu\cdot \epsilon_2, \label{eq::compound_wiretap_thm}
    \end{align}
    \else
    \begin{align}
        & P_e +\nu\cdot \gamma \nonumber \\
        & \leq N_{\epsilon_1}(\mathcal{D}) 
        \sup_{P_{Y|X}\in \mathcal{D}} 
        \mathbf{E}_{Y|X\sim P_{Y|X}} 
        \Big[
        \min\left\{
        \mathsf{L}\mathsf{A}2^{- \hat{\iota}(U;Y)}, 1
        \right\}
        \Big] + \epsilon_1  \label{eq::compound_wiretap_decode} \\
        & \quad + \nu\cdot N_{\epsilon_2}(\mathcal{E}) 
        \bigg(\,
        \sup_{\tilde{P}_{Z|X}\in \mathcal{E}} 
        2 \cdot \mathbf{E}_{Z|X\sim \tilde{P}_{Z|X}} 
        \Big[ 
        \big(1 + 2^{-\iota(U;Z)} \big)^{-\mathsf{B}} 
        \Big]   + \sqrt{ \mathsf{B}\mathsf{A}^{-1}}
        \,\bigg)
        +  \nu\cdot \epsilon_2, \label{eq::compound_wiretap_secrecy}
    \end{align}
    \fi
    where we assume $(U, X, Y, Z)\sim P_{U, X} P_{Y|X} P_{Z|X}$ in the expectation, and $\iota(U;Z), \hat{\iota}(U;Y)$ are the information densities computed for compound channels by the joint distribution $P_{U,X} \hat{P}_{Y|X} P_{Z|X}$, assuming that they are almost surely finite for every $P_{Y|X} \in \mathcal{D}, P_{Z|X} \in \mathcal{E}$. 
\end{theorem}

\begin{proof}

Since the encoder and the decoder are uninformed of the transmission channel, we first design our coding strategy assuming that the transmission channel to the legitimate decoder is fixed to $\hat{P}_{Y|X} \in \mathcal{D}$. 
The eavesdropper is aware of both the transmission channel and the eavesdropping channel $P_{Z|X}\in \mathcal{E}$ that are in use.

Let $\mathcal{C} := ((\bar{U}_i, \bar{M}_i),T_i)_i$ where $(T_i)_i$ is a Poisson process, $\bar{U}_i \stackrel{\mathrm{iid}}{\sim} P_U$, and $\bar{M}_i \stackrel{\mathrm{iid}}{\sim} P_M$ (where $P_M = \mathrm{Unif}[\mathsf{L}]$). 
This is a random codebook that is known to both the encoder and the decoder. 
We first analyze the performance averaged over the random codebook $\mathcal{C}$; and later we will fix the codebook. 
we analyze the error probability and secrecy as follows.

Let $A \sim \mathrm{Unif} [ \mathsf{A}]$ be independent of $(M, \mathcal{C})$. 
The encoder observes the message $M\sim P_M$, computes $U = \tilde{U}_{P_U \times\delta_M}(A)$ by the Poisson functional representation~\cite{li2018strong}, and sends the generated $X|U\sim P_{X|U}$. 
The decoder observes the channel ouptut $Y$ and recovers $\hat{M} = \tilde{M}_{\hat{P}_{U|Y}(\cdot |Y) \times P_M}$ where $\hat{P}_{U|Y}$ is the conditional distribution computed by the joint distribution $P_{U,X}\hat{P}_{Y|X}$.  
We have $(M, A, U, X, Y, Z)\sim P_M\times P_A\times P_{U,X}\hat{P}_{Y|X}\hat{P}_{Z|X}$. 
For the case of a fixed $\hat{P}_{Y|X} \in \mathcal{D}$, the error probability $ \mathbf{P}\big\{ M\neq \hat{M}\big\}$ in this case, denoted by $P_e(\hat{P}_{Y|X})$, can be bounded as follows:
\begin{align}
    & P_e(\hat{P}_{Y|X}) 
    \nonumber \\
    & \leq \mathbf{E}
    \bigg[\mathbf{P}\big((U,M) \neq (\tilde{U}, \tilde{M})_{\hat{P}_{U|Y}(\cdot|Y) \times P_M} | M, A, U, Y \big) \bigg] \nonumber \\
    & \stackrel{(a)}{\leq} \mathbf{E}
    \bigg[\min \Big\{ \mathsf{A} \frac{d P_U \times \delta_M}{d \hat{P}_{U|Y}(\cdot | Y) \times P_M} (U, M), 1  \Big\} \bigg] \nonumber \\
    & = \mathbf{E}
    \bigg[\min \Big\{ \mathsf{L} \mathsf{A} 2^{-\hat{\iota}_{U; Y}(U; Y)} , 1  \Big\} \bigg] \nonumber \\
    & \leq \sup_{P_{Y|X}\in \mathcal{D}} 
    \mathbf{E}_{Y|X \sim A_{Y|X}}\Big[
    \min \Big\{ \mathsf{L} \mathsf{A} 2^{-\hat{\iota}_{U; Y}(U; Y)} , 1  \Big\} 
    \Big] 
    \label{eq::P_e_comp_wiretap} \\
    & =: \overline{P_e}\nonumber 
\end{align}
where $(a)$ is by the conditional
generalized Poisson matching lemma~\cite{li2021unified}, and we define~\eqref{eq::P_e_comp_wiretap} to be $\overline{P_e}$.

For the secrecy measure $\gamma$, note again that the wiretap channel $P_{Z|X}$ is known to the eavesdropper, and for this choice of $P_{Z|X}$, for the total variation distance $\gamma(P_{Z|X})$ we can write:
\begin{align}
    & \mathbf{E} \Big[ \Big\Vert P_{M,Z|\mathcal{C}}(\cdot, \cdot | \mathcal{C}) - P_{M}\times P_{Z|\mathcal{C}}(\cdot | \mathcal{C}) \Big\Vert_{\mathrm{TV}} \Big]  \nonumber \\
    & = \mathbf{E} \Big[
    \Big\Vert
    P_{Z|M, \mathcal{C}}(\cdot, \cdot | \mathcal{C}) -  P_{Z|\mathcal{C}}(\cdot | \mathcal{C})
    \Big\Vert_{\mathrm{TV}}
    \Big] \nonumber \\
    & \leq \mathbf{E}\Big[
    \Big\Vert
    P_{Z|M, \mathcal{C}}(\cdot, \cdot | \mathcal{C}) - P_{Z}(\cdot)
    \Big\Vert_{\mathrm{TV}}
    \Big]   \nonumber \\
    & \qquad \qquad  + \mathbf{E}\Big[
    \Big\Vert
    P_{Z|\mathcal{C}}(\cdot | \mathcal{C}) - P_{Z}(\cdot)
    \Big\Vert_{\mathrm{TV}}
    \Big] \nonumber \\
    & \stackrel{(a)}{\leq} 
    2\cdot\mathbf{E}\Big[
    \Big\Vert
    P_{Z|M, \mathcal{C}}(\cdot, \cdot | \mathcal{C}) - P_{Z}(\cdot)
    \Big\Vert_{\mathrm{TV}}
    \Big]  \nonumber \\
    & = 2 \cdot \mathbf{E}\Big[
    \Big\Vert \mathsf{A}^{-1}
    \sum_{a=1}^\mathsf{A}
    P_{Z|U}(\cdot | \tilde{U}_{P_U\times \delta_M}(a) ) -  P_{Z}(\cdot)
    \Big\Vert_{\mathrm{TV}}
    \Big]  \nonumber \\
    & \stackrel{(b)}{\leq}   
    2\cdot\mathbf{E} 
    \Big[ 
    \big(1 + 2^{-\iota(U;Z)} \big)^{-\mathsf{B}} 
    \Big] 
    + \sqrt{ \mathsf{B}\mathsf{A}^{-1}}
     \nonumber \\
    & \leq \sup_{P_{Z|X}\in \mathcal{E}} 
    2 \cdot \mathbf{E}
    \Big[ 
    \big(1 + 2^{-\hat{\iota}(U;Z)} \big)^{-\mathsf{B}} 
    \Big]
    + \sqrt{ \mathsf{B}\mathsf{A}^{-1}} \label{eq::gamma_comp_wiretap} 
    \\
    & =: \bar{\gamma}
    \nonumber
\end{align}
where $(a)$ is by the convexity of total variation distance, $(b)$ is by~\cite[Proposition 3]{li2021unified} since $\{ \tilde{U}_{P_U\times \delta_m}(a) \}_{a\in[\mathsf{A}]} \stackrel{\mathrm{iid}}{\sim} P_U $ for any $m$, and we define~\eqref{eq::gamma_comp_wiretap} to be $\overline{\gamma}$.

Let $P_e(P_{Y|X}, c)$ denote be the probability of error when the legitimate channel is $P_{Y|X}$ and the codebook is $\mathcal{C} = c$ and also let $\gamma(P_{Z|X}, c)$ denote the total variation distance $\gamma$ when the wiretap channel is $P_{Z|X}$ and the codebook is $\mathcal{C} = c$.

Let $\tilde{\mathcal{D}} \subseteq \mathcal{D}$ attain the minimum in $N_{\epsilon_1}(\mathcal{D})$ and $\tilde{\mathcal{E}} \subseteq \mathcal{E}$ attain the minimum in $N_{\epsilon_2}(\mathcal{E})$. 
Consider any $P_{Y|X} \in \mathcal{D}$ and any $P_{Z|X} \in \mathcal{E}$, and let $\tilde{P}_{Y|X} \in \tilde{\mathcal{D}}, \tilde{P}_{Z|X} \in \tilde{\mathcal{E}}$ satisfy 
\begin{align*}
    & \sup_{x \in \mathcal{X}} 
    \left\Vert 
    P_{Y|X}(\cdot|x) - \tilde{P}_{Y|X}(\cdot|x) 
    \right\Vert_{\mathrm{TV}} 
    \leq \epsilon_1, \\
    & \sup_{x \in \mathcal{X}} 
    \left\Vert 
    P_{Z|X}(\cdot|x) - \tilde{P}_{Z|X}(\cdot|x) 
    \right\Vert_{\mathrm{TV}} 
    \leq \epsilon_2. 
\end{align*} 

The total variation distance between the joint distribution of $M,A,U,X,Y,Z$ under the channel $P_{Y|X}$ (or $P_{Z|X}$) conditional on $\mathcal{C} = c$ and the joint distribution under the channel $\tilde{P}_{Y|X}$ (or $\tilde{P}_{Z|X}$) conditional on $\mathcal{C} = c$ is also bounded by $\epsilon_1$ (or $\epsilon_2$). 
Therefore, we have 
$| P_e(P_{Y|X}, c) - P_e(\tilde{P}_{Y|X}, c) | \leq \epsilon_1$ and
$| \gamma(P_{Z|X}, c) - \gamma(\tilde{P}_{Z|X}, c) | \leq \epsilon_2$. 
Hence, 
\begin{align*}
& P_e(P_{Y|X}, \mathcal{C}) +\nu \cdot \gamma(P_{Z|X}, \mathcal{C}) \\
& \leq P_e \left(\tilde{P}_{Y|X}, c\right) + \epsilon_1 + \nu \cdot \gamma \left(\tilde{P}_{Z|X}, c\right)  + \nu \cdot \epsilon_2 \\
& \leq \sum_{\tilde{P}_{Y|X} \in \tilde{\mathcal{D}} } P_e \left(\tilde{P}_{Y|X}, c\right)+ \epsilon_1  \\
& \qquad \qquad + \nu \cdot \sum_{\tilde{P}_{Z|X} \in \tilde{\mathcal{E}} } \gamma \left(\tilde{P}_{Z|X}, c\right) + \nu \cdot \epsilon_2. 
\end{align*}

Therefore, combining with our previous analyses on the error probability and secrecy, we have
\begin{align*}
& \mathbf{E}_{\mathcal{C}} \Big[
\sup_{P_{Y|X} \in \mathcal{D}, P_{Z|X} \in \mathcal{E}} 
\left(
P_e(P_{Y|X}, \mathcal{C}) +\nu \cdot \gamma(P_{Z|X}, \mathcal{C})
\right) 
\Big]\\
& \leq \mathbf{E}_{\mathcal{C}}\Bigg[
\sum_{\tilde{P}_{Y|X} \in \tilde{\mathcal{D}} } P_e \left(\tilde{P}_{Y|X}, \mathcal{C} \right) + \epsilon_1  \\
& \qquad  \qquad  \qquad  + \nu\cdot \sum_{\tilde{P}_{Z|X} \in \tilde{\mathcal{E}} } \gamma \left(\tilde{P}_{Z|X}, \mathcal{C} \right) + \nu\cdot \epsilon_2 
\Bigg] \\
& \leq |\tilde{\mathcal{D}}| \cdot \overline{P_e} + \nu\cdot |\tilde{\mathcal{E}}| \cdot \overline{\gamma} + \epsilon_1 + \nu\cdot \epsilon_2,
\end{align*} 
and there exists a fixed set of points $((\bar{u}_i, \bar{m}_i), t_i)_{i\in\mathbb{N}}$ such that conditioned on codebook $\mathcal{C} = ((\bar{U}_i, \bar{M}_i),T_i)_{i\in\mathbb{N}} = ((\bar{u}_i, \bar{m}_i), t_i)_{i\in\mathbb{N}}$ the desired bound is attained, and hence the proof is completed.

\end{proof}

\subsection{Recovery of the Asymptotic Results}

In this section, we recover the existing asymptotic results~\cite{liang2009compound}. The procedure is generally similar to Section~\ref{sec::recovery_hiding} where we recover the asymptotic information hiding capacity~\cite{moulin2003information}. 
In~\cite{liang2009compound}, it was assumed that all the random variables are discrete, the channels are memoryless, and $\mathcal{D}_n := \{P_{Y_1|X}^n, \ldots, P_{Y_\mathsf{J}|X}^n \}$ and $\mathcal{E}_n := \{P_{Z_1|X}^n, \ldots, P_{Z_\mathsf{K}|X}^n \}$ for some finite $\mathsf{J}, \mathsf{K}\in\mathbb{N}_+$. 
The setting~\cite{liang2009compound} is depicted in Figure~\ref{fig:wiretap}. 

\begin{figure}[htpb]
	\centering
    \includegraphics[scale = 0.3]
    {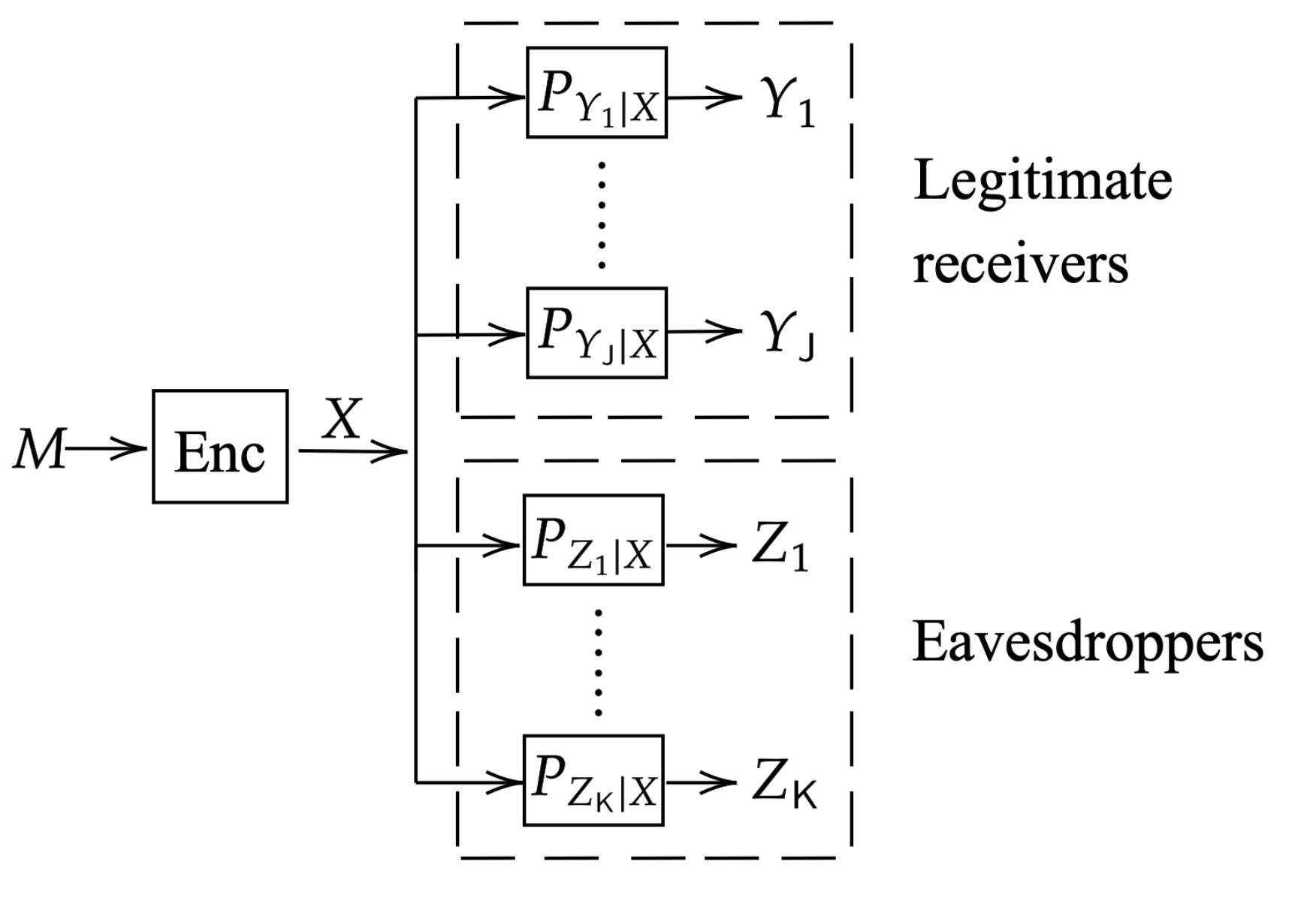}
	\caption{
	Discrete memoryless compound wiretap channel setting from~\cite{liang2009compound}. 
    }
    \label{fig:wiretap}
\end{figure} 

By~\cite{liang2009compound}, for discrete memoryless channels, an achievable secrecy rate is 
\begin{align}
    R & = \max\Big(
    \min_j I(U; Y_j) - \max_k I(U; Z_k)
    \Big) \nonumber \\
    & = \max \min_{j, k} \big(
     I(U; Y_j) - I(U; Z_k)
    \big), \label{eq::asymp_cmpd_wiretap}
\end{align}
where the maximum is taken over all distributions $P_{U,X}$ such that the auxiliary random variable $U$ satisfies the Markov chain $U \leftrightarrow X \leftrightarrow (Y_j, Z_k)$ for $j = 1, \ldots, \mathsf{J}$ and $k = 1, \ldots, \mathsf{K}$. 
Intuitively, this means that we should design the auxiliary random variable $U$ that maximizes the worst-case communication rate $I(U; Y_j) - I(U; Z_k)$, where we consider the \emph{worst} receiver in $\mathcal{D}_n$ and the \emph{most-powerful} eavesdropper in $\mathcal{E}_ n$.

We now show that Theorem~\ref{thm:compound_wiretap} recovers above asymptotic achievable rate in~\eqref{eq::asymp_cmpd_wiretap}, for discrete memoryless compound wiretap channels with finite channel sets and finite $\mathcal{U},\mathcal{X},\mathcal{Y},\mathcal{Z}$. 
We note that although the notations become more complicated, the idea is similar to the recovery of the asymptotic information-hiding capacity in Section~\ref{sec::recovery_hiding}.

Let
\begin{align*}
    \mathcal{D}_n & :=\{ P_{Y_1|X}^n, \ldots,P_{Y_{\mathsf{J}}|X}^n\},\\
    \mathcal{E}_n & :=\{ P_{Z_1|X}^n,\ldots,P_{Z_{\mathsf{K}}|X}^n\},
\end{align*}for some $\mathsf{J},\mathsf{K}\in\mathbb{N}_+$. 
Fix $P_{U,X}$ that achieves
\[
\max_{P_{U,X}} \min_{j,k}\big(I(U;Y_j)-I(U;Z_k)\big),
\]
and write
\begin{align*}
    I_Y &:= \min_{j\in[1:\mathsf{J}]} I(U;Y_j),\\
    I_Z & := \max_{k\in[1:\mathsf{K}]} I(U;Z_k).
\end{align*}
Fix any $\epsilon>0$ and any rate
\[
R < I_Y - I_Z - \epsilon.
\]
We will show that $R$ is achievable. 

Let $\mathsf{L}=\lfloor 2^{nR}\rfloor$. Take $\mathsf{A}:=2^{\lfloor n(I_Z+\epsilon/2)\rfloor}$ and $\mathsf{B}:=2^{\lfloor n(I_Z+\epsilon/3)\rfloor}$, and also take
\begin{equation}
\hat{P}_{Y^n|X^n}
:=\frac{1}{\mathsf{J}}\sum_{j=1}^{\mathsf{J}} P_{Y_j|X}^n .
\label{eq::wiretap_Phat_app}
\end{equation}
We apply Theorem~\ref{thm:compound_wiretap} with $\nu=1$ and $\epsilon_1=\epsilon_2=0$.

Fix any $j_0 \in[1:\mathsf{J}]$ and any $k\in[1:\mathsf{K}]$. Consider the i.i.d. sequences
\[
(U^n,X^n,Y^n,Z^n)\sim P_{U,X}^n\, P_{Y_{j_0}|X}^n\, P_{Z_k|X}^n.
\]
In Theorem~\ref{thm:compound_wiretap}, the term~\eqref{eq::compound_wiretap_decode} corresponds to the decoding error, whereas the term~\eqref{eq::compound_wiretap_secrecy} corresponds to secrecy. 
We first consider the decoding error term~\eqref{eq::compound_wiretap_decode} as follows.

Recall that in Theorem~\ref{thm:compound_wiretap}, the information density $\iota(U^n;Z^n)$ is computed under
$P_{U,X}^n P_{Y_{j_0}|X}^n P_{Z_k|X}^n$, while $\hat{\iota}(U^n;Y^n)$ is computed under
$P_{U,X}^n \hat{P}_{Y^n|X^n} P_{Z_k|X}^n$. 
We denote $P_{Y_j|U}(y|u):=\sum_{x}P_{X|U}(x|u)P_{Y_j|X}(y|x)$ and $P_{Y_j|U}^n(y^n|u^n)=\prod_{i=1}^n P_{Y_j|U}(y_i|u_i)$,  considering $P_{X^n|U^n}=P_{X|U}^n$ and~\eqref{eq::wiretap_Phat_app}, we have
\[
\hat{P}_{Y^n|U^n}(y^n|u^n)
=\sum_{x^n} P_{X^n|U^n}(x^n|u^n)\hat{P}_{Y^n|X^n}(y^n|x^n)
=\frac{1}{\mathsf{J}}\sum_{j=1}^{\mathsf{J}} P_{Y_j|U}^n(y^n|u^n),
\]
and similarly $\hat{P}_{Y^n}(y^n) = {1}/{\mathsf{J}}\cdot\sum_{j=1}^{\mathsf{J}} P_{Y_j}^n(y^n)$. 
Therefore,
\begin{equation}
\hat{\iota}(u^n;y^n)
:=\log\frac{\hat{P}_{Y^n|U^n}(y^n|u^n)}{\hat{P}_{Y^n}(y^n)}
=
\log\frac{{1}/{\mathsf{J}}\cdot \sum_{j=1}^{\mathsf{J}}P_{Y_j|U}^n(y^n|u^n)}
{{1}/{\mathsf{J}} \,\cdot \sum_{j=1}^{\mathsf{J}}P_{Y_j}^n(y^n)} .
\label{eq::hat_iota_nletter_def}
\end{equation}

{\color{black}

For each $j\in[1 : \mathsf{J}]$, denote $P_{Y_j}(y):=\sum_{u}P_U(u)P_{Y_j|U}(y|u)$ and $P_{Y_j}^n(y^n):=\prod_{i=1}^n P_{Y_j}(y_i)$, and
\[
\iota_j(u^n;y^n)
:=
\log\frac{P_{Y_j|U}^n(y^n|u^n)}{P_{Y_j}^n(y^n)}.
\]
For the chosen channel $P_{Y_{j_0}|X}$, under
$(U^n,Y^n)\sim P_U^nP_{Y_{j_0}|U}^n$, we know $\mathbf{E}[\iota_{j_0}(U^n;Y^n)] = nI(U;Y_{j_0})$. }
For $j\neq j_0$, define the following ratio with the convention that it is set to $0$ whenever the denominator is $0$,
\[
    r_{j} := \frac{P_{Y_j}^n(Y^n)}{P_{Y_{j_0}}^n(Y^n)} .
\]
Then under $Y^n\sim P_{Y_{j_0}}^n$, we have
\[
\mathbf{E}[r_{j}]
=
\sum_{y^n:\, P_{Y_{j_0}}^n(y^n)>0}
P_{Y_{j_0}}^n(y^n)\frac{P_{Y_j}^n(y^n)}{P_{Y_{j_0}}^n(y^n)}
\leq 1.
\]
Then by Markov's inequality, for any $\delta>0$, we have
\[
\mathbf{P}\big(r_{j}>2^{n\delta}\big)\leq 2^{-n\delta}.
\]
By a union bound over $j\neq j_0$, with probability at least
$1-(\mathsf{J}-1)2^{-n\delta}$,
\begin{equation}
\sum_{j\neq j_0} r_{j} \leq (\mathsf{J}-1)2^{n\delta}.
\label{eq::mix_good_event_fixed}
\end{equation}
{
\color{black}
Using~\eqref{eq::hat_iota_nletter_def}, we have
\begin{align}
\hat{\iota}(U^n;Y^n)
&=
\log\frac{\sum_{j=1}^{\mathsf{J}}P_{Y_j|U}^n(Y^n|U^n)}
{\sum_{j=1}^{\mathsf{J}}P_{Y_j}^n(Y^n)}
\nonumber\\
&\stackrel{(a)}{\geq}
\log\frac{P_{Y_{j_0}|U}^n(Y^n|U^n)}
{P_{Y_{j_0}}^n(Y^n)\Big(1+\sum_{j\neq j_0}r_{j}\Big)}
\nonumber\\
&=
\iota_{j_0}(U^n;Y^n)
-
\log\Big(1+\sum_{j\neq j_0}r_{j}\Big)
\nonumber\\
&\stackrel{(b)}{\geq}
\iota_{j_0}(U^n;Y^n)
-
\log \big(1+(\mathsf{J}-1)2^{n\delta}\big)
\nonumber\\
&\stackrel{(c)}{\geq}
\iota_{j_0}(U^n;Y^n)-(n\delta+\log\mathsf{J}).
\label{eq::hat_lower_bound_fixed}
\end{align}
where $(a)$ is from $\sum_{j=1}^{\mathsf{J}}P_{Y_j|U}^n(Y^n|U^n)\geq P_{Y_{j_0}|U}^n(Y^n|U^n)$ and $\sum_{j=1}^{\mathsf{J}}P_{Y_j}^n(Y^n) = P_{Y_{j_0}}^n(Y^n) \big(1+\sum_{j\neq j_0}r_{j}\big)$; $(b)$ is from~\eqref{eq::mix_good_event_fixed}; and $(c)$ is from $1 + (\mathsf{J}-1)2^{n\delta}\leq \mathsf{J}2^{n\delta}$.
}

Next, since $\mathcal{U},\mathcal{Y}$ are finite, by Hoeffding's inequality, for any $\delta>0$,
\begin{equation}
\mathbf{P}\big(\iota_{j_0}(U^n;Y^n)\leq n(I(U;Y_{j_0})-\delta)\big)
\leq t_n.
\label{eq::iy_conc_fixed}
\end{equation}
where $t_n\leq \exp(-c n\delta^2)$ for some constant $c > 0$ independent of $j_0$, and hence $t_n$ tends to $0$ exponentially fast with respect to $n$.

Let $\mathcal{G}_{j_0}$ denote the intersection of the event~\eqref{eq::mix_good_event_fixed} and the event
$\{\iota_{j_0}(U^n;Y^n)\geq n(I(U;Y_{j_0})-\delta)\}$. Combining~\eqref{eq::hat_lower_bound_fixed} and~\eqref{eq::iy_conc_fixed},
on $\mathcal{G}_{j_0}$ we have
\[
\hat{\iota}(U^n;Y^n)\geq n(I(U;Y_{j_0})-2\delta)-\log\mathsf{J},
\]
and moreover
\[
\mathbf{P}(\mathcal{G}_{j_0}^c)\leq (\mathsf{J}-1)2^{-n\delta}+t_n.
\]
Therefore, 
\begin{align}
&\mathbf{E}_{Y^n|X^n\sim P_{Y_{j_0}|X}^n}\Big[\min\big\{\mathsf{L}\mathsf{A}2^{-\hat{\iota}(U^n;Y^n)},1\big\}\Big] \nonumber \\
&\leq \mathbf{E}\big[\mathsf{L}\mathsf{A}2^{-\hat{\iota}(U^n;Y^n)}\mathbf{1}_{\mathcal{G}_{j_0}}\big]
+\mathbf{P}(\mathcal{G}_{j_0}^c) \nonumber \\
&\leq \mathsf{L}\mathsf{A}\cdot 2^{-n(I(U;Y_{j_0})-2\delta)}\cdot 2^{\log\mathsf{J}}
+(\mathsf{J}-1)2^{-n\delta}
+ t_n \nonumber \\
&\leq \mathsf{J}\cdot 2^{n(R+I_Z+\epsilon/2-I(U;Y_{j_0})+2\delta)}
+(\mathsf{J}-1)2^{-n\delta}
+ t_n. \label{eq::decoding_term_final}
\end{align}
Choose $\delta=\epsilon/8$. Since $R<I_Y-I_Z-\epsilon$ and $I(U;Y_{j_0})\geq I_Y$, we have
\begin{align*}
    & R+I_Z+\epsilon/2-I(U;Y_{j_0})+2\delta \\
    & \leq (I_Y-I_Z-\epsilon)+I_Z+\epsilon/2-I_Y+\epsilon/4 \\
    & = -\epsilon/4,
\end{align*}
so the right-hand side of~\eqref{eq::decoding_term_final} vanishes exponentially as $n\to\infty$, uniformly over $j_0$. 
Recall $\mathcal{D}_n$ is finite, and $N_{0}(\mathcal{D}_n)\leq \mathsf{J}$, consequently,
\begin{equation}
N_{0}(\mathcal{D}_n)  \cdot \sup_{P_{Y|X}^n\in\mathcal{D}_n}
\mathbf{E}_{Y^n|X^n\sim P_{Y|X}^n}
\Big[\min\big\{\mathsf{L}\mathsf{A}2^{-\hat{\iota}(U^n;Y^n)},1\big\}\Big]
\to 0.
\label{eq::decoder_term_vanish_fixed}
\end{equation}

We then consider the secrecy term~\eqref{eq::compound_wiretap_secrecy}.
Let $P_{Z_k|U}(z|u):=\sum_x P_{X|U}(x|u)P_{Z_k|X}(z|x)$ and $P_{Z_k}(z):=\sum_u P_U(u)P_{Z_k|U}(z|u)$. 
Fix any $k\in[1:\mathsf{K}]$ and define
\[
\iota_k(u^n;z^n):=\log\frac{P_{Z_k|U}^n(z^n|u^n)}{P_{Z_k}^n(z^n)}
=\sum_{i=1}^n \log\frac{P_{Z_k|U}(z_i|u_i)}{P_{Z_k}(z_i)}.
\]
We have $\mathbf{E}[\iota_k(U^n;Z^n)]=nI(U;Z_k)\leq nI_Z$. 
Since $\mathcal{U},\mathcal{Z}$ are finite, by Hoeffding's inequality, for any $\eta>0$,
\begin{equation}
\mathbf{P}\big(\iota_k(U^n;Z^n)\geq n(I(U;Z_k)+\eta)\big)
\leq s_n, 
\label{eq::iz_conc_fixed}
\end{equation}
where $s_n\leq \exp(-c n\eta^2)$ for some constant $c>0$ independent of $k$, and hence $s_n$ tends to $0$ exponentially fast with respect to $n$.

Define the event $\mathcal{T}_k:=\{\iota_k(U^n;Z^n)\leq n(I(U;Z_k)+\eta)\}$. Then by~\eqref{eq::iz_conc_fixed}, 
\begin{align}
\mathbf{E}\Big[\big(1+2^{-\iota_k(U^n;Z^n)}\big)^{-\mathsf{B}}\Big]
&\leq \mathbf{P}(\mathcal{T}_k^c)
+\big(1+2^{-n(I(U;Z_k)+\eta)}\big)^{-\mathsf{B}} \nonumber \\
&\stackrel{(a)}{\leq} s_n 
+\exp \Big(-(\mathsf{B}/2)\cdot 2^{-n(I(U;Z_k)+\eta)}\Big), \label{eq::secrecy_term_bound_fixed}
\end{align}
where the inequality $(a)$ uses $\ln(1+t)\geq t/2$ for $t\in(0,1]$ and the fact that $2^{-n(I(U;Z_k)+\eta)}\leq 1$.
With $\mathsf{B}= 2^{\lfloor n(I_Z+\epsilon/3)\rfloor}$ and $I(U;Z_k)\leq I_Z$, taking $\eta=\epsilon/6$ gives
\[
(\mathsf{B}/2)\cdot 2^{-n(I(U;Z_k)+\eta)}
\geq \frac{1}{4}\cdot 2^{n(I_Z+\epsilon/3-I_Z-\epsilon/6)}
=\frac{1}{4}\cdot2^{n\epsilon/6}, 
\]
i.e., the second term of~\eqref{eq::secrecy_term_bound_fixed} tends to $0$ super-exponentially fast. 
Hence
\begin{equation}
\sup_{P_{Z|X}^n\in\mathcal{E}_n}
\mathbf{E}_{Z^n|X^n\sim P_{Z|X}^n}
\Big[\big(1+2^{-\iota(U^n;Z^n)}\big)^{-\mathsf{B}}\Big]
\to 0
\label{eq::secrecy_exp_vanish_fixed}
\end{equation}
exponentially fast, and the other term in~\eqref{eq::compound_wiretap_secrecy}, 
\begin{equation}
\sqrt{\mathsf{B}\mathsf{A}^{-1}}
\leq
\sqrt{2}\cdot 2^{{n}/{2}(I_Z+\epsilon/3-I_Z-\epsilon/2)}
=
\sqrt{2}\cdot 2^{-n\epsilon/12}\to 0,
\label{eq::BA_vanish_fixed}
\end{equation}
also exponentially  fast. 
Recall $\mathcal{E}_n$ is finite, and $N_0(\mathcal{E}_n)\leq \mathsf{K}$, therefore we conclude
\begin{align}
	N_{0} (\mathcal{E}_n) \bigg(\, \sup_{\tilde{P}_{Z|X}^n\in \mathcal{E}_n} 2 \cdot \mathbf{E}_{Z^n|X^n\sim \tilde{P}_{Z|X}^n} \Big[ \big(1 + 2^{-\iota(U^n;Z^n) } \big)^{-\mathsf{B}} \Big]   + \sqrt{ \mathsf{B}\mathsf{A}^{-1}} \,\bigg) \to 0 \label{eq::secrecy_final2}
\end{align}

Combining~\eqref{eq::decoder_term_vanish_fixed} and~\eqref{eq::secrecy_final2}, we conclude that for the terms in Theorem~\ref{thm:compound_wiretap}, $P_e+\gamma\to 0$ as $n\to\infty$, and therefore any rate $R<I_Y-I_Z$ is achievable.
Maximizing over $P_{U,X}$ completes the proof and recovers~\eqref{eq::asymp_cmpd_wiretap}.

\section{Future Applications}
\label{sec::watermarking_AI}

We briefly discuss some future directions related to the information hiding problem. 

In recent years, with the great success of artificial intelligence, software based on generative models is now able to produce content as realistic as works by human creators. 
However, the issue of copyright has become controversial: generative models can potentially be trained on public data without obtaining permission from human authors. 
To protect intellectual property and prohibit the unauthorized use of original works in the training of generative models, practical watermarking tools have been developed~\cite{shan2023glaze, zhu2018hidden, zhang2021brief} to embed information into human-created works before they are published. 
However, most existing implementations rely on ad hoc designs that may introduce vulnerabilities~\cite{ji2024principled}. 
For better design of such tools, it is crucial to develop a solid understanding of the information hiding problem. 
As an example, a principled design for text watermarking was proposed in~\cite{ji2024principled}, following the theoretical analysis of the information hiding problem~\cite{moulin2003information}. 
Since our one-shot achievability results generalize the asymptotic information hiding, drop certain assumptions and still recover the asymptotic hiding capacity~\cite{moulin2003information} (thus providing an alternative understanding), it is of interests to explore how can our nonasymptotic information hiding results contribute to improved designs for watermarking technologies.

Moreover, steganography is closely related to the information hiding problem~\cite{wang2008perfectly}, and an interesting connection between steganography and minimum entropy coupling~\cite{li2020efficient, cicalese2019minimum, compton2023minimum} has been discussed in~\cite{de2022perfectly}. 
Considering that the coding scheme in this paper is based on the Poisson functional representation~\cite{li2018strong}, which is also related to the minimum entropy coupling (see discussions in~\cite{li2020efficient}), it is worth investigating whether there are deeper connections between minimum entropy coupling and the one-shot information hiding problem studied in this paper.

\section{Conclusion}

In this paper, we derive novel one-shot achievability results for two fundamental secrecy problems in information theory: the information hiding problem~\cite{moulin2003information} and the compound wiretap channel~\cite{liang2009compound}. 
Our bounds are based on the Poisson matching lemma combined with a covering argument, and they apply to both continuous and discrete settings. 
For both problems, unlike most existing studies, we do not assume that the decoder knows the channel state in the one-shot setting. 
The proposed one-shot results hold for arbitrary source distributions and general channel classes (not necessarily memoryless or ergodic), and they recover existing asymptotic results when applied to discrete memoryless channels, potentially under distortion constraints. 
Thus, our results offer alternative, and potentially simpler, proofs of known capacity theorems. 
We have also discussed possible generalizations and future applications.

\bibliographystyle{IEEEtran}
\bibliography{ref.bib}
\end{document}